\documentclass[conference]{IEEEtran}
\usepackage{eso-pic} 
\usepackage{algorithm}
\usepackage{algpseudocode}
\usepackage{amsmath}
\usepackage{amssymb}
\usepackage{url}
\usepackage{epsfig}
\usepackage{graphicx}
\usepackage{caption}
\usepackage{subcaption}

\makeatletter
\renewcommand{\ALG@beginalgorithmic}{\small}
\makeatother

\algnotext{EndIf}
\algnotext{EndFor}
\algnotext{EndWhile}

\newtheorem{theorem}{Theorem}[section]
\newtheorem{definition}{Definition}[section]
\newtheorem{remark}{Remark}[section]
\newtheorem{example}{Example}[section]

\newtheorem{lemma}[theorem]{Lemma}
\newtheorem{proposition}[theorem]{Proposition}

\newenvironment{proof}[1][Proof]{\begin{trivlist}
\item[\hskip \labelsep {\bfseries #1}]}{\end{trivlist}}

\newcommand{\qed}{\nobreak \ifvmode \relax \else
      \ifdim\lastskip<1.5em \hskip-\lastskip
      \hskip1.5em plus0em minus0.5em \fi \nobreak
      \vrule height0.75em width0.5em depth0.25em\fi}
      
\newcounter{MYtempeqncnt}      

\ifCLASSINFOpdf
\else
\fi
\begin{document}
%
\title{Anonymizing Social Graphs via Uncertainty Semantics}

\author{\IEEEauthorblockN{Hiep H. Nguyen, Abdessamad Imine, and Micha\"{e}l Rusinowitch}
\IEEEauthorblockA{LORIA/INRIA Nancy-Grand Est, France}
Email: \{huu-hiep.nguyen,michael.rusinowitch\}@inria.fr, abdessamad.imine@loria.fr
}


%


\maketitle

\begin{abstract}
Rather than anonymizing social graphs by generalizing them to super nodes/edges or adding/removing nodes and edges to satisfy given privacy parameters, recent methods exploit the semantics of uncertain graphs to achieve privacy protection of participating entities and their relationship. These techniques anonymize a deterministic graph by converting it into an uncertain form. In this paper, we propose a generalized obfuscation model based on uncertain adjacency matrices that keep expected node degrees equal to those in the unanonymized graph. We analyze two recently proposed schemes and show their fitting into the model. We also point out disadvantages in each method and present several elegant techniques to fill the gap between them. Finally, to support fair comparisons, we develop a new tradeoff quantifying framework by leveraging the concept of incorrectness in location privacy research. Experiments on large social graphs demonstrate the effectiveness of our schemes.
\end{abstract}



%
\IEEEpeerreviewmaketitle

\section{Introduction}
\label{sec:introduction}

Graphs represent a rich class of data observed in daily life where entities are described by vertices and their connections are characterized by edges. With the emergence of increasingly complex networks \cite{newman2003structure}, the research community requires large and reliable graph data to conduct in-depth studies. However, this requirement usually conflicts with privacy rights of data contributing entities. Naive approaches like removing user ids from a social graph are not effective, leaving users open to privacy risks, especially re-identification attacks \cite{backstrom2007wherefore} \cite{hay2008resisting}. Therefore, many graph anonymization schemes have been proposed \cite{zhou2008preserving,liu2008towards,zou2009k,cheng2010k,wu2010k,tai2011privacy}. 

Given an unlabeled undirected graph, the existing anonymization methods fall into four main categories. The first category includes \textit{random} addition, deletion and switching of edges to prevent the re-identification of nodes or edges. The methods in the second category provide k-anonymity \cite{sweeney2002k} by \textit{deterministic} edge additions or deletions, assuming attacker's background knowledge regarding certain properties of its target nodes. The methods in the third category assign edge probabilities to add uncertainty to the true graph. The edges probabilities may be computed explicitly as in \cite{boldi2012injecting} or implicitly via random walks \cite{mittal2012preserving}. Finally, the fourth class of techniques, \textit{generalization}, cluster nodes into super nodes of size at least \textit{k}. Note that the last two classes of schemes induce \textit{possible world} models, i.e., we can retrieve sample graphs that are consistent with the anonymized output graph.

The third category is the most recent class of methods which leverage the semantics of edge probability to inject uncertainty to a given deterministic graph, converting it into an uncertain one. Most of schemes in this category are scalable, i.e. runnable on million-scale graphs or more. As an example, Boldi et al. \cite{boldi2012injecting} introduced the concept of \textit{(k,$\epsilon$)-obfuscation} (denoted as ($k,\epsilon$)-obf), where $k \geq 1$ is a desired level of obfuscation and $\epsilon \geq 0$ is a tolerance parameter. However, the pursuit for minimum standard deviation $\sigma$ in (k,$\epsilon$)-obf has high impact on node privacy and high privacy-utility tradeoff. Edge rewiring method based on random walks (denoted as \textit{RandWalk}) in \cite{mittal2012preserving} also introduces uncertainty to edges as we show in section \ref{sec:uni-model}. This scheme suffers from high lower bounds for utility despite its excellent privacy-utility tradeoff. 

Motivated by \textit{(k,$\epsilon$)-obf} and \textit{RandWalk}, we propose in this work a generalized model for anonymizing graphs based on edge uncertainty. Both (k,$\epsilon$)-obf and RandWalk display their fitting into the model. We point out disadvantages in (k,$\epsilon$)-obf and RandWalk, the tradeoff gap between them and present several elegant techniques to fill this gap. Finally, to support fair comparisons, we develop a new tradeoff quantifying framework using the concept of \textit{incorrectness} in location privacy research \cite{shokri2011quantifying}.

Our contributions are summarized as follows:
\begin{itemize}
\item We propose a generalized model called \textit{uncertain adjacency matrix} for anonymizing graph via edge uncertainty semantics (Section \ref{sec:uni-model}). The key property of this model is that expected degrees of all nodes must be unchanged. We show the fitting of \textit{(k,$\epsilon$)-obf} and \textit{RandWalk} into the model and then analyze their disadvantages (Sections \ref{sec:pre}, \ref{sec:uni-model}). 
\item We introduce the \textit{Maximum Variance} (MaxVar) scheme (Section \ref{sec:max-var}) that satisfies all the properties of the uncertain adjacency matrix. It achieves good privacy-utility tradeoff by using two key observations: nearby potential edges and maximization of total node degree variance via a simple quadratic program.
\item Towards a fair comparison for anonymization schemes on graphs, this paper describes a generic quantifying framework (Section \ref{sec:quantify}) by putting forward the distortion measure (also called \textit{incorrectness} in \cite{shokri2011quantifying}) to measure the re-identification risks of nodes. As for the utility score, typical graph metrics \cite{boldi2012injecting} \cite{ying2008randomizing} are chosen. 
\item We conduct a comparative study of aforementioned approaches on three real large graphs and show the effectiveness of our gap-filling solutions (Section \ref{sec:eval}).
\end{itemize}
 
 Table \ref{tab:notation} summarizes notations used in this paper.
 
 \begin{table}[htb]
 \small
 \centering
 \caption{List of notations} \label{tab:notation}
 \begin{tabular}{|c|l|}
 \hline
 \textbf{Symbol} &\textbf{Definition} \\
 \hline
 $G_0=(V,E_{G_0})$ & true graph with $n=|V|$ and $m=|E_{G_0}|$\\
 \hline
 $\mathcal{G}=(V,E,p)$ & uncertain graph constructed from $G_0$\\
 \hline
 $G=(V,E_G)$ & sample graph from $\mathcal{G}$, $G \sqsubseteq \mathcal{G}$ \\
 \hline
 $d_u(G), d_u(\mathcal{G})$ & degree of node $u$ in $G, \mathcal{G}$ \\
 \hline
 $\Delta(d)$ & number of nodes having degree $d$ in $G$ \\
 \hline
 $\mathcal{N}(u)$ & neighbors of node $u$ in $\mathcal{G}$  \\
 \hline
 $R_\sigma$ & truncated normal distribution on [0,1] \\
 \hline
 $r_e \leftarrow R_\sigma$ & a sample from the distribution $R_\sigma$\\  
 \hline
 $p_i$ ($p_{uv}$) & probability of edge $e_i$ ($e_{uv}$)\\
 \hline
 $n_p$ & number of potential edges, $|E|=m+n_p$ \\
 \hline
 $A$, $\mathcal{A}$ & adjacency matrices of $G_0$, $\mathcal{G}$\\
 \hline
 $P_{RW}$ & random walk transition matrix of $G_0$\\
 \hline
 $B^{(t)}$ & uncertain adjacency matrix, $B^{(t)} = AP_{RW}^{t-1}$\\
 \hline
 $t$ & walk length\\
 \hline
 $S$ & switching matrix\\  
 \hline
 $TV$ & total degree variance\\   
 \hline

 \end{tabular}
 \end{table}
 
\section{Related Work}
\label{sec:related}

\subsection{Anonymizing Deterministic Graphs}
There is a vast literature on graph perturbation that deserves a survey. In this section, we enumerate only several groups of ideas that are related to our proposed schemes.
\subsubsection{Anonymizing unlabeled vertices for node privacy}
In unlabeled graphs, node identifiers are numbered in an arbitrary manner after removing their labels. An attacker aims at reidentifying nodes solely based on their structural information. For this line of graphs, node privacy protection implies link privacy. Techniques of adding and removing edges, nodes can be done randomly or deterministically. Random perturbation is a naive approach and usually used as a baseline method. More guided approaches consist of \textit{k-neighborhood}\cite{zhou2008preserving}, \textit{k-degree}\cite{liu2008towards}, \textit{k-automorphism}\cite{zou2009k}, \textit{k-symmetry}\cite{wu2010k}, \textit{k-isomorphism}\cite{cheng2010k} and \textit{$k^2$-degree}\cite{tai2011privacy}. These schemes provide k-anonymity \cite{sweeney2002k} semantics and usually rely on heuristics to avoid combinatorial intractability. K-automorphism, k-symmetry, and k-isomorphism can resist \textit{any} structural attacks by exploiting the inherent symmetry in graph. $k^2$-degree addresses the friendship attacks, based on the vertex degree pair of an edge. Ying and Wu \cite{ying2008randomizing} propose a spectrum preserving approach which wisely chooses edge pairs to switch in order to keep the spectrum of the adjacency matrix not to vary too much. The clearest disadvantage of the above schemes is that they are inefficient on large scale graphs.

Apart from the two above categories, perturbation techniques have other categories that settle on \textit{possible world} semantics. Hay et al. \cite{hay2008resisting} generalize a network by clustering nodes and publish graph summarization of super nodes and super edges. The utility of this scheme is limited. On the other hand, Boldi et al. \cite{boldi2012injecting} take the uncertain graph approach. With edge probabilities, the output graph can be used to generate sample graphs by independent edge sampling. Our approach belongs to this class of techniques with different formulation and better privacy-utility tradeoff. Note that in \textit{k-symmetry}\cite{cheng2010k}, the output sample graphs are also possible worlds of the symmetric intermediate graph.

\subsubsection{Anonymizing labeled vertices for link privacy}
If nodes are labeled, we are only concerned about the link disclosure risk. For example, Mittal et al. \cite{mittal2012preserving} employ an edge rewiring method based on random walks to keep the mixing time tunable and prevent link re-identification by Bayesian inference. This method is effective for social network based systems, e.g. Sybil defense, DHT routing. Link privacy is also described in \cite{ying2008randomizing} for Random Switch, Random Add/Del. Interestingly, \textit{RandWalk} \cite{mittal2012preserving} can also be used for unlabeled graphs as shown in Section \ref{sec:uni-model}.

\subsubsection{Min entropy, Shannon entropy and incorrectness measure}
We now survey commonly used notions of privacy metrics. \textit{Min entropy} \cite{smith2009foundations} quantifies the largest probability gap between the posterior and prior over all items in the input dataset. \textit{K-anonymity} has the same semantics with the corresponding min entropy of $\log_2k$. So we say k-anonymity based perturbation schemes belong to min entropy. Shannon entropy argued in \cite{bonchi2011identity} and \cite{boldi2012injecting} is another choice of privacy metrics. The third metrics that we use in this paper is the \textit{incorrectness} measure from location privacy \cite{shokri2011quantifying}. Given the prior information (e.g. node degree in the true graph) and the posterior information harvested from the anonymized output, incorrectness measure is the number of incorrect guesses made by the attacker. This measure gauges the \textit{distortion} caused by the anonymization algorithm.

\subsection{Mining Uncertain Graphs}
Uncertain graphs pose big challenges to traditional mining techniques. Because of the exponential number of possible worlds, naive enumerations are intractable. Typical graph search operations like k-Nearest neighbor and pattern matching require new approaches \cite{potamias2010k} \cite{zou2010mining} \cite{yuan2011efficient}. Those methods answer threshold-based queries by using pruning strategies based on Apriori property of frequent patterns.

\section{Preliminaries}
\label{sec:pre}
This section starts with definitions and common assumptions on uncertain graphs. It then analyzes vulnerabilities in $(k,\epsilon)$-obf \cite{boldi2012injecting}.

\subsection{Uncertain Graph}
\label{subsec:uncertain-graph}
Let $\mathcal{G} = (V,E,p)$ be an uncertain undirected graph, where $p: E \rightarrow [0,1]$ is the function that gives an existence probability to each edge (see Fig.\ref{fig:ke-graph-obf}). The common assumption is on the \textit{independence} of edge probabilities. Following the \textit{possible-worlds} semantics in relational data \cite{dalvi2007management}, the uncertain graph $\mathcal{G}$ induces a set \{$G=(V,E_G)$\} of $2^{|E|}$ deterministic graphs (worlds), each is defined by a subset of $E$. The probability of $G=(V,E_G) \sqsubseteq \mathcal{G}$ is:

\begin{equation}
Pr(G) = \prod_{e \in E_G} p(e) \prod_{e \in E \setminus E_G} (1 - p(e))
\end{equation}
Note that deterministic graphs are also uncertain graphs with all edges having probabilities 1.

\subsection{$(k,\epsilon)$-obf and Its Limitations} 
\label{subsec:ke-obf}

In \cite{boldi2012injecting}, Boldi et al. extend the concept of \textit{k-obfuscation} developed earlier \cite{bonchi2011identity}.

\begin{definition}
\label{def:ke-obf}
\textit{(k,$\epsilon$)-obf} \cite{boldi2012injecting}. Let P be a vertex property, $k \geq 1$ be a desired level of obfuscation, and $\epsilon \geq 0$ be a tolerance parameter. The uncertain graph $\mathcal{G}$ is said to k-obfuscate a given vertex $v \in G$ with respect to P if the entropy of the distribution $Y_{P(v)}$ over the vertices of $\mathcal{G}$ is greater than or equal to $\log_2k$:\\
\begin{equation}
H(Y_{P(v)}) \geq \log_2k 	\label{eqn:ke}
\end{equation}
The uncertain graph $\mathcal{G}$ is a \textit{$(k,\epsilon)$-obf} with respect to property P if it \textit{k-obfuscates} at least $(1-\epsilon)n$ vertices in $\mathcal{G}$ with respect to P. \hfill $\Box$
\end{definition}

\begin{figure}
	\begin{center}
        \begin{subfigure}[b]{0.13\textwidth}
                \centering
                \epsfig{file=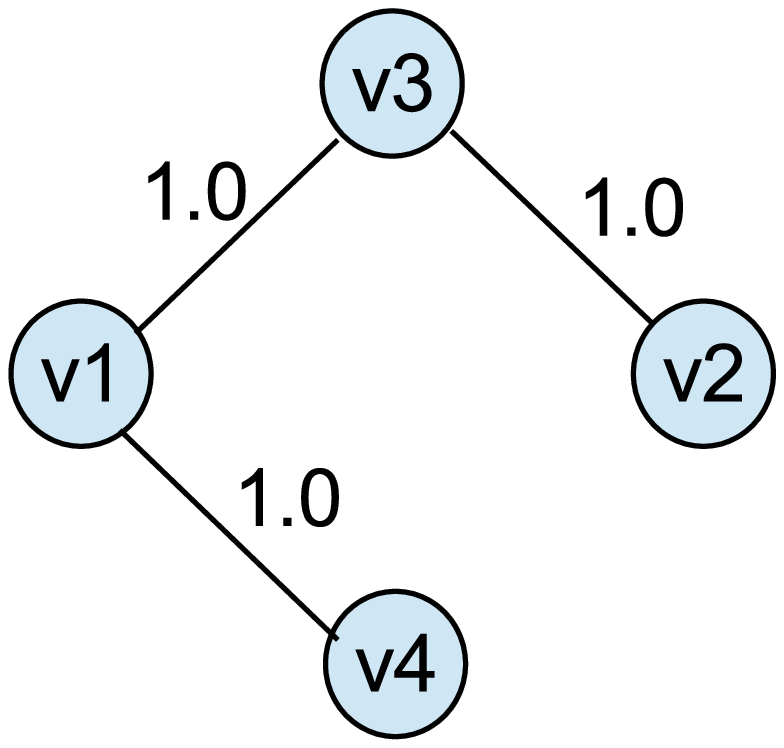, height=1.0in}
                \setlength{\abovecaptionskip}{-10pt}
                \caption{}	
                \label{fig:ke-graph}
        \end{subfigure}
		\begin{subfigure}[b]{0.13\textwidth}
                \centering
                \epsfig{file=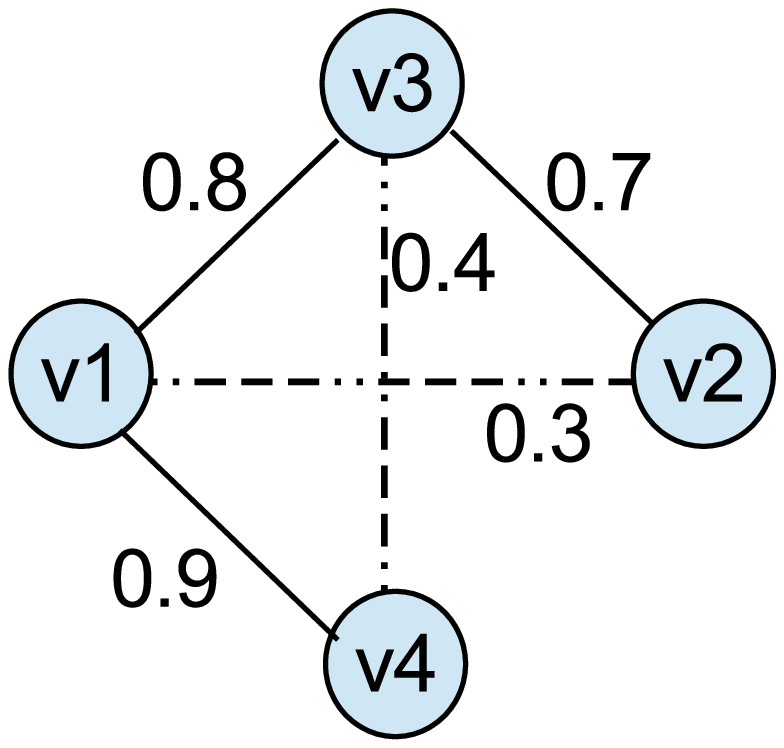, height=1.0in}
                \setlength{\abovecaptionskip}{-10pt}
                \caption{}	
                \label{fig:ke-graph-obf}
        \end{subfigure}        
        \begin{subfigure}[b]{0.2\textwidth}
                \epsfig{file=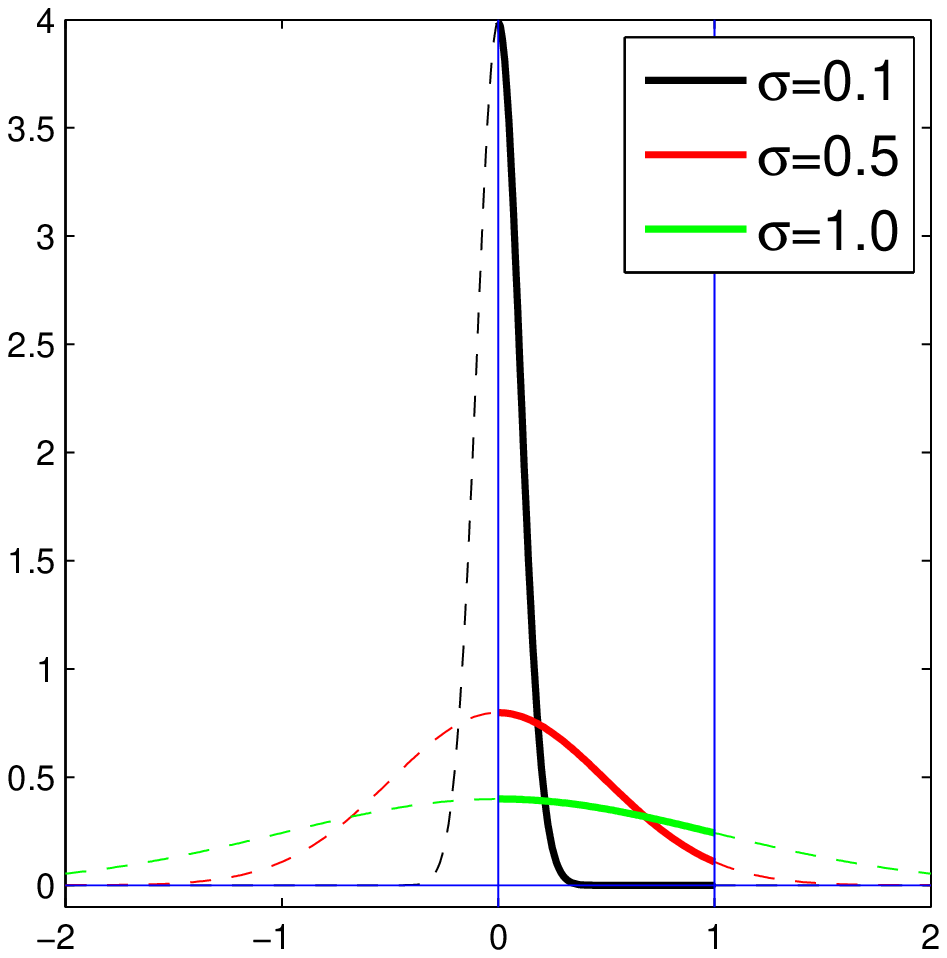, height=1.3in}
                \setlength{\abovecaptionskip}{-10pt}
                \caption{}	
                \label{fig:truncated-normal}
        \end{subfigure}        
    \end{center}	
    \setlength{\abovecaptionskip}{-5pt}
    \caption{(a) True graph (b) An obfuscation with potential edges (dashed) (c) Truncated normal distribution on [0,1] (bold solid curves)}
    \vspace{-1.5em}
    \label{fig:ke-obf}    
\end{figure}

Given the true graph $G_0$ (Fig.\ref{fig:ke-graph}), the basic idea of \textit{$(k,\epsilon)$-obf} (Fig.\ref{fig:ke-graph-obf}) is to transfer the probabilities from existing edges to \textit{potential} (non-existing) edges to satisfy Definition \ref{def:ke-obf}. 
For each existing sampled edge $e$, it is assigned a probability $1-r_e$ where $r_{e} \leftarrow R_\sigma$ (Fig. \ref{fig:truncated-normal}) and for each non-existing sampled edge $e'$, it is assigned a probability $r_{e'} \leftarrow R_\sigma$. 

Table \ref{tab:XY} gives an example of how to compute degree entropy for the uncertain graph in Fig. \ref{fig:ke-graph-obf}. Here vertex property $P$ is the node degree. Each row in the left side is the degree distribution for the corresponding node. For instance, $v1$ has degree $0$ with probability $(1-0.8).(1-0.3).(1-0.9) = 0.014$. The right side normalizes values in each column (i.e. in each degree value) to get distributions $Y_{P(v)}$. The entropy $H(Y_{P(v)})$ for each degree value is shown in the bottom row. Given $k=3, \log_2k=1.585$, then $v1, v3$ with true degree 2 and $v2, v4$ with true degree 1 satisfy (\ref{eqn:ke}). Therefore, $\epsilon = 0$.

\begin{table}[htb]
\small
\centering
\caption{The degree uncertainty for each node (left) and normalized values for each degree (right)} \label{tab:XY}
\begin{tabular}{|r|r|r|r|r|r|r|r|r|}
\hline
& \multicolumn{4}{c|}{node degree uncertainty} & \multicolumn{4}{c|}{$Y_{P(v)}$} \\ 
\hline
 &d=0 &d=1 &d=2 &d=3 & d=0 &d=1 &d=2 &d=3 \\
\hline
v1 & .014 & .188 & .582 & .216 & .044 & .117 & .355 & .491 \\
v2 & .210 & .580 & .210 & .000 & .656 & .362 & .128 & .000 \\
v3 & .036 & .252 & .488 & .224 & .112 & .158 & .298 & .509 \\
v4 & .060 & .580 & .360 & .000 & .187 & .362 & .220 & .000 \\
\hline
\multicolumn{4}{|c|}{} & $H$ & 1.40 & 1.84 & 1.91 & 0.99 \\ 
\hline
\end{tabular}
\end{table}

While \textit{$(k,\epsilon)$-obf} provides a novel technique to come up with an uncertain version of the graph, the specific approach in \cite{boldi2012injecting} has two drawbacks. First, it formulated the problem as the minimization of $\sigma$. With small values of $\sigma$, $r_e$ highly concentrates around zero, so existing sampled edges have probabilities nearly 1 and non-existing sampled edges are assigned probabilities almost 0. By the simple rounding technique, the attacker can easily reveal the true graph. Even if the graph owner only publishes sample graphs, the re-identification attacks are still effective as we show in Section \ref{sec:eval}. Note that in \cite{boldi2012injecting}, the found values of $\sigma$ vary in a wide range from $10^{-1}$ to $10^{-8}$. Second, the approach in \cite{boldi2012injecting} does not consider the locality (subgraph) of nodes in selecting pairs of nodes for establishing potential edges. As shown in \cite{fard2012limiting}, \textit{subgraph-wise perturbation} effectively reduces structural distortion. 


\section{A Generalized Model for Uncertain Graph}
\label{sec:uni-model}
This section introduces a generalized model of graph anonymization via semantics of edge uncertainty. Then we analyze several schemes using this model.

\subsection{A Generalized Model: Uncertain Adjacency Matrix}
\label{subsec:model-overview}

\begin{figure}
	\begin{center}
        \begin{subfigure}[b]{0.33\textwidth}
                \centering
                \includegraphics[height=1.2in]{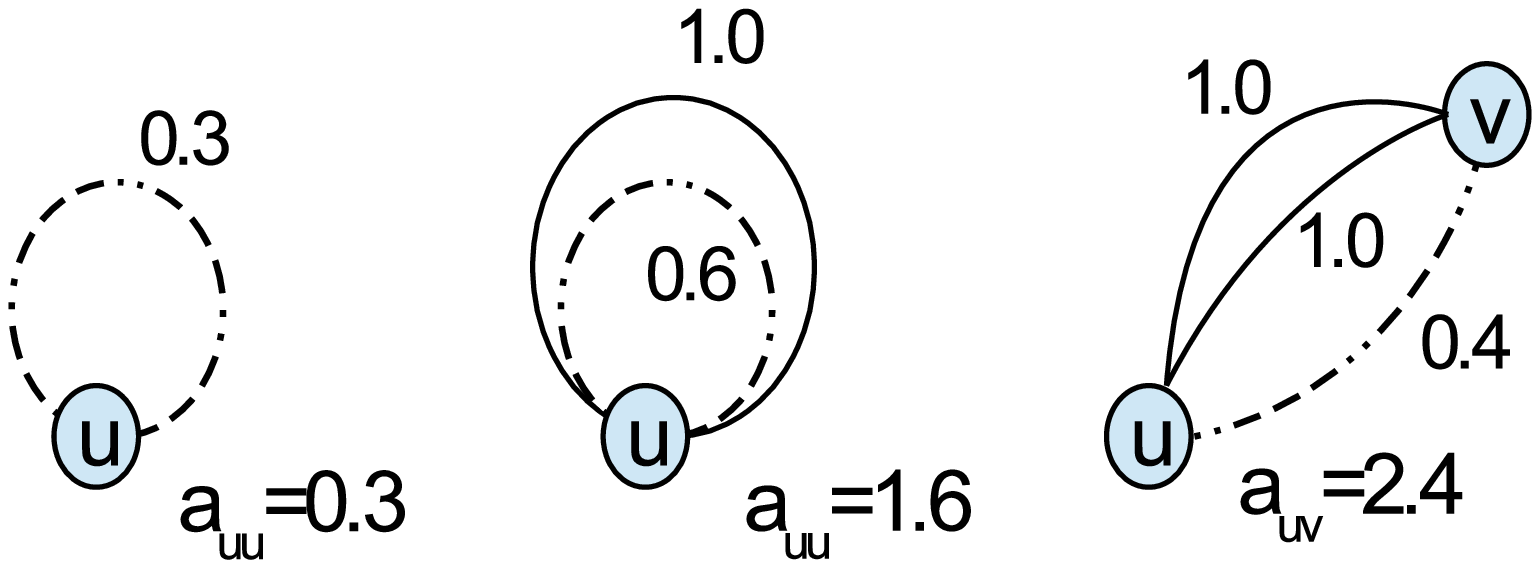}
                \setlength{\abovecaptionskip}{-10pt}
                \caption{}	
                \label{fig:selfloop-multiedge}
        \end{subfigure}
		\begin{subfigure}[b]{0.14\textwidth}
                \centering
                \includegraphics[height=1.2in]{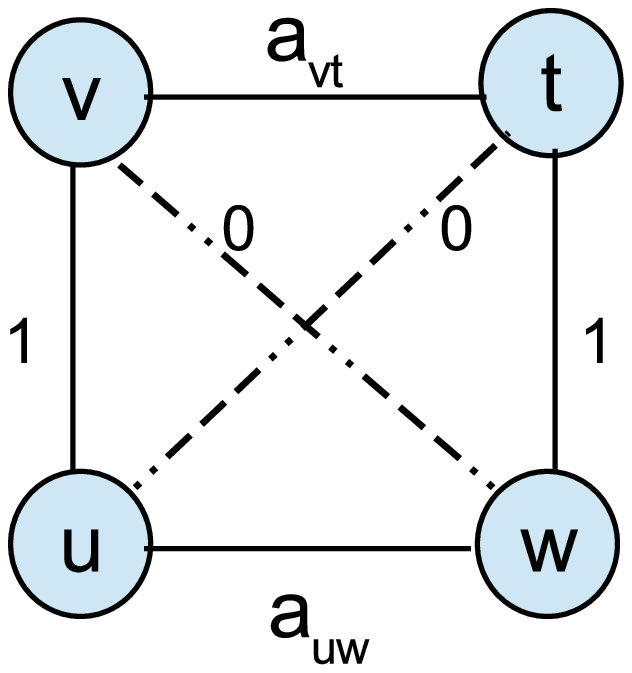}
                \setlength{\abovecaptionskip}{-20pt}
                \caption{}	
                \label{fig:switch}
        \end{subfigure}
    \end{center}	
    \setlength{\abovecaptionskip}{-5pt}
    \caption{(a) Semantics of selfloops (left), multi-selfloops (middle) and multiedges (right) in uncertain adjacency matrix (b) Edge switching}
    \vspace{-1.5em}
\end{figure}

Given the true graph $G_0$, an uncertain graph $\mathcal{G}$ constructed from $G_0$ must have its uncertain adjacency matrix $\mathcal{A}$ satisfying
\begin{enumerate}
\item symmetry $\mathcal{A}_{ij}=\mathcal{A}_{ji}$
\item $\mathcal{A}_{ij} \in [0,1]$ and $\mathcal{A}_{ii}=0$. If we relax this constraint to \hfill \\
(2') allow $\mathcal{A}_{ii} > 0$ then we have \textit{selfloop}s and allow $\mathcal{A}_{ij} > 1$ then we have \textit{multiedge}s (Fig. \ref{fig:selfloop-multiedge}).
\item \textit{expected degrees} of all nodes must be unchanged. It means $\sum_{j=1}^{n} \mathcal{A}_{ij} = d_i(G_0)	 \;\; i=1..n$
\end{enumerate}

We first define the transition matrix $P_{RW}$ which is right \textit{stochastic} (i.e. non-negative and row sums equal to 1) as follows (note that we use the short notation $d_i = d_i(G_0)$)

\begin{equation}
 P_{RW}(i,j) = \begin{cases}
        1/d_i  & \text{if } (i,j) \in E_{G_0} \;\; i \ne j\\
        0 & \text{otherwise}.
        \end{cases}
\end{equation}


The power $P_{RW}^t$ when $t \rightarrow \infty$ is $P_{RW}^{\infty}(i,j) = \frac{d_j}{2m}$.


We prove two lemmas on properties of the products $\mathcal{A}P$ and $AP^t$ where $P$ is right stochastic.

\begin{lemma}
\label{lem:expected-deg}
For an adjacency matrix $\mathcal{A}$ and a right stochastic matrix $P$, the product $\mathcal{A}P$ is non-negative and has row sums equal to those of $\mathcal{A}$.
\end{lemma}
\begin{proof}
The non-negativity of $\mathcal{A}P$ is trivial. The sum of row $i$ of $\mathcal{A}P$ is $\sum_{j}(\sum_{k}\mathcal{A}_{ik} P_{kj}) = \sum_{k}\mathcal{A}_{ik}(\sum_{j}P_{kj}) = \sum_{k}\mathcal{A}_{ik}.1 = \sum_{k}\mathcal{A}_{ik}$
\end{proof}

\begin{lemma}
\label{lem:symmetry}
For a deterministic graph $G$ possessing adjacency matrix $A$ and $P_{RW}$, the product $B^{(t)}=AP_{RW}^{t-1}$ is also symmetric.
\end{lemma}
\begin{proof}
We prove the result by induction. The case $t=1$ is trivial. We prove that for any $t \geq 2$, $B^{(t)}_{ij} = \sum_{p_t(i,j)} \prod_{k \in p_t(i,j), k\ne i,j} 1/d_k$ where $p_t(i,j)$ is a path of length $t$ from $i$ to $j$. 

When $t=2$, $B^{(2)}_{ij} = \sum_{k} A_{ik}P_{kj} = \sum_{(i,k),(k,j) \in E} 1/d_k$, so the result holds. Assuming that the result is correct up to $t-1$, i.e. $B^{(t-1)}_{ij} = \sum_{p_{t-1}(i,j)} \prod_{k \in p_{t-1}(i,j), k\ne i,j} 1/d_k$. Because $B^{(t)} = B^{(t-1)} P_{RW}$, $B^{(t)}_{ij} = \sum_{l} B^{(t-1)}_{il} P_{lj} = \sum_{l,(l,j)\in E} (\sum_{p_{t-1}(i,l)} \prod_{k \in p_{t-1}(i,l), k\ne i,l} 1/d_k)1/d_l =  \sum_{p_t(i,j)} \prod_{k \in p_t(i,j), k\ne i,j} 1/d_k$.

Because $G$ is undirected, the set of all $p_t(i,j)$ is equal to the set of all $p_t(j,i)$, so $B^{(t)}_{ij}=B^{(t)}_{ji}$.
%
\end{proof}

We prove the uniqueness of $P_{RW}$ in the following proposition.

\begin{proposition}
\label{propos:uniqueness}
Given a deterministic graph $G$ with adjacency matrix $A$, there exists one and only one right stochastic matrix $P$ that satisfies $P_{uv} = 0$ for all $(u,v) \notin G$ and $AP^t$ is symmetric for all $t \geq 0$. The unique solution is $P=P_{RW}$.
\end{proposition}
\begin{proof}
Lemma \ref{lem:symmetry} shows that $P=P_{RW}$ satisfies $P_{uv} = 0$ for all $(u,v) \notin G$ and $AP^t$ is symmetric for all $t\geq 0$. 

To prove that this is the unique solution, we repeat the formula in the proof of Lemma \ref{lem:symmetry}. Let $B^{(t)} = AP^{t-1}$, then $B^{(t)}_{ij} = \sum_{p_t(i,j)} \prod_{k \in p_t(i,j), k\ne i,j} P_{k,k+1}$ where $k+1$ implies the successive node of $k$ in $p_t(i,j)$. Because $B^{(t)}_{ji}$ has the same number of products as $B^{(t)}_{ij}$ (i.e. the number of paths of length $t$), $B^{(t)}$ is symmetric if and only if corresponding products are equal, i.e. $\prod_{k \in p_t(i,j), k\ne i,j} P_{k,k+1} =  \prod_{k \in p_t(j,i), k\ne i,j} P_{k,k+1}$. At $t=2$, for any path $(i,k,j)$ we must have $P_{kj} = P_{ki}$. Along with the requirement that $P$ is right stochastic, i.e. $\sum_{i}P_{ki}=1$, we obtain $P_{ki}=1/d_k$. This is exactly $P_{RW}$.
\end{proof}

\subsection{RandWalk Approach}
\label{subsec:model-randwalk}

Now we apply the model of uncertain adjacency matrix to the analysis of \textit{RandWalk} \cite{mittal2012preserving}. Algorithm \ref{algo-randwalk} depicts the steps of RandWalk. As we show below, the trial-and-error condition in Line 6 makes RandWalk hard to analyze \footnote{\label{foot:1}It also causes edge miss at $t=2$, e.g. a 2-length walk on edge $(v3,v2)$ (Fig. \ref{fig:ke-graph}) causes the selfloop $(v3,v3)$.}. 
So we modify it by removing the condition and using parameter $\alpha$ instead of 1.0 in Line 12 \footnote{This line causes errors for degree-1 nodes as shown in \textit{RandWalk-mod}.} (see Algorithm \ref{algo-randwalk-mod}). When $\alpha=0.5$, all edges $(u,z)$ are assigned with probability 0.5. In RandWalk-mod, we add a checking for $d_u=1$ (Line 8) to keep the total degree of $G'$ equal to that of $G$, which is missing in RandWalk. Note that RandWalk-mod accepts selfloops and multiedges.

Let $Q$ be the \textit{edge adding} matrix defined as 
\begin{equation}
Q_{ij} = \begin{cases}
		0.5 & \text{if } d_i = 1 \wedge j \text{ is the unique neighbor of } i \\
        \alpha  & \text{if } j \text{ is the first neighbor of } i \\
        \frac{0.5d_i - \alpha}{d_i-1} & \text{if } j \text{ is a neighbor of } i \text{ but not the first one}\\
        0 & \text{otherwise}.
        \end{cases} \nonumber
\end{equation}

We show that \textit{RandWalk-mod} can be formulated as an uncertain adjacency matrix $\mathcal{A}_{RW} = (AP_{RW}^{t-1}) \circ (Q + Q^T)$, where $\circ$ is the Hadamard product (element-wise). $AP_{RW}^{t-1}$ is equivalent to computations in lines 2-6 and $Q + Q^T$ is equivalent to computations in lines 7-13. We use $Q + Q^T$ instead of $Q$ due to the fact that when the edge $(u,z)$ is added to $G'$ with probability $Q_{uz}$, the edge $(z,u)$ is also assigned the same probability. We come up with the following theorem.

\begin{theorem}
\label{theorem:randwalk-mod}
\textit{RandWalk-mod} can be formulated as  $\mathcal{A}_{RW} = (AP_{RW}^{t-1}) \circ (Q + Q^T)$. $\mathcal{A}_{RW}$ is symmetric. It satisfies the constraint of unchanged expected degree iff $\alpha=0.5$ \footnote{This implies a mistake in Theorem 3 of \cite{mittal2012preserving}}.
\end{theorem}
\begin{proof}
By Lemmas \ref{lem:expected-deg} and \ref{lem:symmetry}, let $B^{(t)}_{RW}$ be $AP_{RW}^{t-1}$, we have symmetric $B^{(t)}_{RW}$ and its row sums are equal to those of $A$. Because $\mathcal{A}_{RW} = B^{(t)}_{RW} \circ (Q + Q^T)$ and both $B^{(t)}_{RW}$ and $(Q + Q^T)$ are symmetric, $\mathcal{A}_{RW}$ is also symmetric.

Due to the fact that $(Q + Q^T)$ has the same locations of non-zeros as $B^{(t)}_{RW}$, the condition of unchanged expected degree is satisfied if and only if all non-zeros in $(Q + Q^T)$ are 1. This occurs if and only if  $\alpha=0.5$.
\end{proof}

\begin{algorithm}
\caption{RandWalk($G_0,t,M$) \cite{mittal2012preserving}}
\label{algo-randwalk}
\begin{algorithmic}[1]
	\Require undirected graph $G_0$, walk length $t$ and maximum loop count $M$
	\Ensure anonymized graph $G'$
	\State $G'=null$
	\For {$u$ in $G_0$}
		\State $count=1$
		\For {$v$ in $\mathcal{N}(u)$}
			\State $loop=1$
			\While {$(u==z \vee (u,z)\in G') \wedge (loop \leq M)$}
				\State perform $t-1$ hop random walk from $v$
				\State $z$ is the terminal node of the random walk
				\State $loop++$
			\EndWhile
			\If {$loop \leq M$}
				\If {$count==1$}
					\State add $(u,z)$ to $G'$ with probability 1.0 
				\Else
					\State add $(u,z)$ to $G'$ with probability $\frac{0.5d_u - 1}{d_u-1}$
				\EndIf
			\EndIf
			\State $count++$
		\EndFor
	\EndFor
	\Return $G'$
\end{algorithmic}
\end{algorithm}

\begin{algorithm}
\caption{RandWalk-mod($G_0,t,\alpha$)}
\label{algo-randwalk-mod}
\begin{algorithmic}[1]
	\Require undirected graph $G_0$, walk length $t$ and probability $\alpha$
	\Ensure anonymized graph $G'$
	\State $G'=null$
	\For {$u$ in $G_0$}
		\State $count=1$
		\For {$v$ in $\mathcal{N}(u)$}
			\State perform $t-1$ hop random walk from $v$
			\State $z$ is the terminal node of the random walk

			\If {$count==1$}
				\If {$d_u==1$}
					\State add $(u,z)$ to $G'$ with probability 0.5
				\Else
					\State add $(u,z)$ to $G'$ with probability $\alpha$
				\EndIf
			\Else
				\State add $(u,z)$ to $G'$ with probability $\frac{0.5d_u - \alpha}{d_u-1}$
			\EndIf
			\State $count++$
		\EndFor
	\EndFor
	\Return $G'$
\end{algorithmic}
\end{algorithm}

We investigate the limit case when $t \rightarrow \infty$ (i.e. $P_{RW}^{t-1} \rightarrow P_{RW}^{\infty}$). Correspondingly $B_{RW}^{\infty} = A P_{RW}^{\infty}$ has $B_{RW}^{\infty}(i,j) = \frac{d_i d_j}{2m}$. 
The following theorem quantifies the number of selfloops and multiedges in $B_{RW}^{\infty}$ for power-law (PL) graphs and sparse Erd\"{o}s-Renyi (ER) random graphs \cite{newman2003structure}.

\begin{theorem}
\label{theorem:limit}
For power-law graphs with the exponent $\gamma$, the number of selfloops in $B_{RW}^{\infty}$ is $\frac{\zeta(\gamma-2)}{\zeta(\gamma-1)}$, where $\zeta(\gamma)$ is the Riemann zeta function defined only for $\gamma > 1$; the number of multiedges is zero.

For sparse ER random graphs with $\lambda=np$ constant where $p$ is the edge probability, the number of selfloops in $B_{RW}^{\infty}$ is $\lambda + 1$; the number of multiedges is zero.
\end{theorem}
\begin{proof} See Appendix \ref{proof:limit}.
\end{proof}

\begin{remark}
We notice that \textit{RandWalk-mod} can be done equivalently by the idea in SybilGuard \cite{yu2006sybilguard}. We first pick a random permutation $\pi_u$ on neighbors of each node $u$ to get $d_u$ pairs of (in-edge, out-edge). Then for any walk reaching node $u$ by the in-edge $(v,u)$, the out-edge is fixed to $(u,\pi_u(v))$. In this formulation, it is straightforward to verify that the transition probability from $u$ to a neighbor $v$ is $1/d_u(G_0)$.
\end{remark}

\subsection{Edge Switching}
\label{subsec:model-switch}
In edge switching (\textit{EdgeSwitch}) approaches (Fig. \ref{fig:switch}), two edges $(u,v), (w,t)$ are chosen and switched to $(u,t), (w,v)$ if $a_{ut} = a_{wv} = 0$. This is done in $s$ switches. Using the switching matrix $S$, we represent 1-step EdgeSwitch in the form $AS=\mathcal{A}$ (Equation (\ref{eqn:edge-switching})).

The switching matrix $S$ is feasible if and only if $a_{uw}a_{vt} = 0$. Note that in the full form, $S$ is $n\times n$ matrix with the $n-4$ remaining elements on diagonal are 1, other off-diagonal are 0. In general, $S$ is not right stochastic and this happens only when $a_{uw} = a_{vt} = 0$. For $s$-step EdgeSwitch $A\prod_{i=1}^{s} S_i=\mathcal{A}$. If $\forall i, S_i$ is right stochastic (i.e. we choose edges $(u,v),(w,t)$ such that $a_{uw} = a_{vt} = 0$), then Lemma \ref{lem:expected-deg} applies.

\begin{figure*}[!t]
\normalsize
\setcounter{MYtempeqncnt}{\value{equation}}
\setcounter{equation}{3}
\begin{equation}
\label{eqn:edge-switching}
\begin{bmatrix}
0 & 1 & a_{uw} & 0\\
1 & 0 & 0 & a_{vt}\\
a_{uw} & 0 & 0 & 1\\
0 & a_{vt} & 1 & 0
\end{bmatrix}
\begin{bmatrix}
0 & -a_{vt} & 1 & a_{vt}\\
-a_{uw} & 0 & a_{uw} & 1\\
1 & a_{vt} & 0 & -a_{vt}\\
a_{uw} & 1 & -a_{uw} & 0
\end{bmatrix}
=
\begin{bmatrix}
0 & 0 & a_{uw} & 1\\
0 & 0 & 1 & a_{vt}\\
a_{uw} & 1 & 0 & 0\\
1 & a_{vt} & 0 & 0
\end{bmatrix}
\end{equation}
\setcounter{equation}{4}
\hrulefill
\vspace*{2pt}
\end{figure*}

\subsection{Direct Construction}
\label{subsec:model-direct}
Given the deterministic adjacency matrix $A$, we can directly construct $\mathcal{A}$ that satisfies all three constraints (1),(2) and (3) in Section \ref{subsec:model-overview}. \textit{(k,$\epsilon$)-obf} \cite{boldi2012injecting} introduces such an approach. As explained in Section \ref{subsec:ke-obf}, the expected degrees of nodes in $(k,\epsilon)$-obf are \textit{approximately} unchanged due to the fact that $r_e, r_{e'}$ are nearly zero by small $\sigma$. So (k,$\epsilon$)-obf satisfies constraints (1) and (2) but it only approximately satisfies the third constraint.

To remedy this shortcoming, we present the MaxVar approach in Section \ref{sec:max-var}. It adds potential edges to $G_0$, then tries to find the assignment of edge probabilities such that the expected node degrees are unchanged while the total variance is maximized. A comparison among schemes is also shown in the end of Section \ref{subsec:compare-model}.

\subsection{Mixture Approach}
\label{subsec:model-mixture}
In this section, we present the \textit{Mixture} approach by the uncertain adjacency matrix $\mathcal{A}_p$ parametrized by $p$, with the output sample graph $G_p$. Given the true graph $G_0$ and an anonymized $G \sqsubseteq \mathcal{G}$, every edge $(i,j)$ is chosen into $G_p$ with probability $\mathcal{A}_p(i,j)$ where

\begin{equation}
 \mathcal{A}_p(i,j) = \begin{cases}
        1 	& \text{if } (i,j) \in E_{G_0} \cap E_{G}\\
        1-p & \text{if } (i,j) \in E_{G_0} \setminus E_{G}\\
        p 	& \text{if } (i,j) \in E_{G} \setminus E_{G_0}
        \end{cases} \nonumber
\end{equation}

It is straightforward to show that $\mathcal{A}_p = (1-p)A(G_0) + pA(G)$. When applied to $G$ generated by \textit{RandWalk-mod} with $\alpha=0.5$, we have $\mathcal{A}_p = (1-p)A + pAP_{RW}^{t-1} = A[(1-p)I_n + pP_{RW}^{t-1}]$ and $\mathcal{A}_p$ satisfies three constraints (1) (2') and (3).

If there exists $P_{mix}$ with constraint $P_{mix}(i,j) = 0 \text{ if } (i,j) \notin E_{G_0}$ such that $P_{mix}^{t-1} = (1-p)I_n + pP_{RW}^{t-1}$, then Mixture can be simulated by the RandWalk-mod approach with the transition matrix $P_{mix}$.

\subsection{Partition Approach}
Another approach that can apply to RandWalk-mod, $(k,\epsilon)$-obf, MaxVar and EdgeSwitch is the \textit{Partition} approach. Given true graph $G_0$, this divide-and-conquer strategy first partitions $G_0$ into disjoint subgraphs $sG$, then it applies one of the above anonymization schemes on subgraphs to get anonymized subgraphs $s\mathcal{G}$. Finally, it combines $s\mathcal{G}$ to obtain $\mathcal{G}$. Note that the partitioning may cause orphan edges as in MaxVar (Section \ref{sec:max-var}). Those edges must be copied to $\mathcal{G}$ to keep node degrees unchanged.

\section{Maximum Variance Approach}
\label{sec:max-var}
We start this section with the formulation of \textit{MaxVar} in the form of quadratic programming based on two key observations. Then we describe the anonymization algorithm.

\subsection{Formulation}
Two key observations underpinning the MaxVar approach are presented as follows.
\subsubsection{Observation \#1: Maximum Degree Variance}
We argue that efficient countermeasures against structural attacks should hinge on node degrees. If a node and its neighbors have their degrees changed, the re-identification risk is reduced significantly. Consequently, instead of replicating local structures as in k-anonymity based approaches \cite{zhou2008preserving,liu2008towards,zou2009k,cheng2010k,wu2010k,tai2011privacy}, we can deviate the attacks by changing node degrees \textit{probabilistically}. For example, node \textit{v1} in Fig.\ref{fig:ke-graph} has degree 2 with probability 1.0 whereas in  Fig.\ref{fig:ke-graph-obf}, its degree gets four possible values $\{0,1,2,3\}$ with probabilities $\{0.014, 0.188, 0.582, 0.216\}$ respectively. Generally, given edge probabilities of node $u$ as $p_1, p_2,..p_{d_u(\mathcal{G})}$, the degree of $u$ is a sum of independent Bernoulli random variables, so its expected value is $\sum_{i=1}^{d_u(\mathcal{G})} p_i$ and its variance is $\sum_{i=1}^{d_u(\mathcal{G})} p_i(1-p_i)$. If we naively target the maximum (local) degree variance without any constraints, the naive solution is at $p_i = 0.5 \;\forall i$. However, such an assignment distorts graph structure severely and deteriorates the utility. Instead, by following the model of uncertain adjacency matrix, we have the constraint $\sum_{i=1}^{d_u(\mathcal{G})} p_i = d_u(G_0)$. Note that the \textit{minimum variance} of an uncertain graph is 0 and corresponds to the case $\mathcal{G}$ has all edges being deterministic, e.g. when $\mathcal{G} = G_0$ and in switching-edge based approaches. In the following section, we show an interesting result relating the \textit{total} degree variance with the variance of edit distance.

\subsubsection{Variance with edit distance}
The \textit{edit distance} between two deterministic graphs $G, G'$ is defined as:

\begin{equation}
D(G,G') = |E_G\setminus E_{G'}| + |E_{G'}\setminus E_G|
\end{equation}

A well-known result about the expected edit distance between the uncertain graph $\mathcal{G}$ and the deterministic graph $G \sqsubseteq \mathcal{G}$ is

\begin{equation}
E[D(\mathcal{G},G)] = \sum_{G' \sqsubseteq \mathcal{G}} Pr(G')D(G,G') = \sum_{e_i \in E_G} (1-p_i) + \sum_{e_i \notin E_G} p_i \nonumber
\end{equation}

Correspondingly, the variance of edit distance is
\begin{equation}
Var[D(\mathcal{G},G)] = \sum_{G' \sqsubseteq \mathcal{G}} Pr(G')[D(G,G')-E[D(\mathcal{G},G)]]^2 \nonumber
\end{equation}

We prove in the following theorem that the variance of edit distance is the sum of all edges' variance (total degree variance) and it does not depend on the choice of $G$.

\begin{theorem}
\label{theorem:edit-distance}
Assume that $\mathcal{G}(V,E,p)$ has $k$ uncertain edges $e_1,e_2,...,e_k$ and $G  \sqsubseteq \mathcal{G}$ (i.e. $E_G \subseteq E$). The edit distance variance is $Var[D(\mathcal{G},G)] = \sum_{i=1}^{k} p_i(1-p_i)$ and does not depend on the choice of $G$.
\end{theorem}
\begin{proof} See Appendix \ref{proof:edit-distance}.
\end{proof}

%
%

\subsubsection{Observation \#2: Nearby Potential Edges}
\label{subsubsec:ob-2}
As indicated by Leskovec et al. \cite{leskovec2007graph}, real graphs reveal two temporal evolution properties: \textit{densification power law} and \textit{shrinking diameters}. Community Guided Attachment (CGA) model \cite{leskovec2007graph}, which produces densifying graphs, is an example of a hierarchical graph generation model in which the linkage probability between nodes decreases as a function of their relative distance in the hierarchy. With regard to this observation, $(k,\epsilon)$-obf, by heuristically making potential edges solely based on node degree discrepancy, produces many inter-community edges. Shortest-path based statistics will be reduced due to these edges. MaxVar, in contrast, tries to mitigate the structural distortion by proposing only \textit{nearby} potential edges before assigning edge probabilities. Another evidence is from \cite{vazquez2003growing} where Vazquez analytically proved that \textit{Nearest Neighbor} can explain the power-law for degree distribution, clustering coefficient and average degree among the neighbors. Those properties are in very good agreement with the observations made for social graphs. Sala et al. \cite{sala2010measurement} confirmed the consistency of Nearest Neighbor model in their comparative study on graph models for social networks.

\subsection{Algorithms}
\label{subsec:max-var-algo}
This section describes the steps of \textit{MaxVar} to convert the input deterministic graph into an uncertain one. 

\subsubsection{Overview}
The intuition behind the new approach is to formulate the perturbation problem as a \textit{quadratic programming} problem. Given the true graph $G_0$ and the number of potential edges allowed to be added $n_p$, the scheme has three phases. The first phase tries to partition $G_0$ into $s$ subgraphs, each one with $n_s = n_p/s$ potential edges connecting nearby nodes (with default distance 2, i.e. \textit{friend-of-friend}). The second phase formulates a quadratic program for each subgraph with the constraint of unchanged node degrees to produce the uncertain subgraphs $s\mathcal{G}$ with maximum edge variance. The third phase combines the uncertain subgraphs $s\mathcal{G}$ into $\mathcal{G}$ and publishes several sample graphs. The three phases are illustrated in Fig. \ref{fig:max-var}. 

By keeping the degrees of nodes in the perturbed graph, our approach is similar to the \textit{edge switching} approaches (e.g.\cite{ying2008randomizing}) but ours is more subtle as we do it implicitly and the switching occurs not necessarily on pairs of edges.

\begin{figure}
\centering
\includegraphics[height=2.0in]{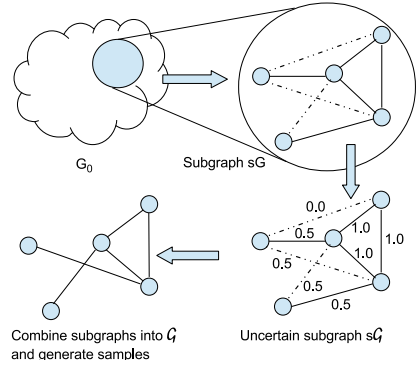}
\vspace{-1em}
\caption{MaxVar approach}
\vspace{-1.5em}
\label{fig:max-var}
\end{figure}

\subsubsection{Graph Partitioning}
Because of the complexity of exact quadratic programming (Section \ref{subsec:quad-prog}), we need a pre-processing phase to divide the true graph $G_0$ into subgraphs and run the optimization on each subgraph. Given the number of subgraphs $s$, we run \textit{METIS} \footnote{http://glaros.dtc.umn.edu/gkhome/views/metis} to get almost equal-sized subgraphs with minimum number of inter-subgraph edges. Each subgraph has $n_s$ potential edges added before running the quadratic program. This phase is outlined in Algorithm \ref{algo-partitioning}.

\begin{algorithm}
\caption{Partition-and-Add-Edges}
\label{algo-partitioning}
\begin{algorithmic}[1]
	\Require true graph $G_0=(V,E_{G_0})$, number of subgraphs $s$, number of potential edges per subgraph $n_s$
	\Ensure list of augmented subgraphs $gl$
	\State $gl \leftarrow$ \texttt{METIS}($G_0, s$).
	\For {$sG$ in $gl$}
		\State $i \leftarrow 0$
		\While {$i < n_s$}
			\State randomly pick $u,v \in V_{sG}$ and $(u,v) \notin E_{sG}$ with $d(u,v)=2$
			\State $E_{sG} \leftarrow E_{sG} \cup (u,v)$
			\State $i \leftarrow i+1$
		\EndWhile
	\EndFor
	\Return $gl$
\end{algorithmic}
\end{algorithm}

\subsubsection{Quadratic Programming}
\label{subsec:quad-prog}
By assuming the independence of edges, the total degree variance of $\mathcal{G} = (V,E,p)$ for edit distance (Theorem \ref{theorem:edit-distance}) is: 

\begin{equation}
\label{eqn:var}
Var(E) = \sum_{i=1}^{|E|} p_i(1-p_i) = |E_{G_0}| - \sum_{i=1}^{|E|} p_i^2
\end{equation}

The last equality in (\ref{eqn:var}) is due to the constraint that the expected node degrees are unchanged (i.e. $\sum_{i=1}^{d_u(\mathcal{G})} p_i = d_u(G_0)$), so $\sum_{i=1}^{|E|} p_i$ is equal to $|E_{G_0}|$. By targeting the maximum edge variance, we come up with the following quadratic program.

\begin{equation*}
\begin{aligned}
& \text{Minimize} & \sum_{i=1}^{|E|} p_i^2 \\
& \text{Subject to} & 0 \leq p_i \leq 1 \; \forall i \\
& & \sum_{v \in \mathcal{N}(u)} p_{uv} = d_u(G_0) \; \forall u
\end{aligned}
\end{equation*}

The objective function reflects the privacy goal (i.e. sample graphs do not highly concentrate around the true graph) while the expected degree constraints aim to preserve the utility.

By dividing the large input graph into subgraphs, we solve independent quadratic optimization problems. Because each edge belongs to at most one subgraph and the expected node degrees in each subgraph are unchanged, it is straightforward to show that the expected node degrees in $\mathcal{G}$ are also unchanged. We have a proposition on problem feasibility and an upper bound for the total variance.

\begin{proposition}
\label{propos:max-var}
The quadratic program in \textit{MaxVar} is always feasible. The total variance $TV_{MaxVar}=Var(E)$ is upper bounded by $\frac{mn_p}{m+n_p}$.
\end{proposition}
\begin{proof}
The feasibility is due to the fact that $\{p_e|p_e = 1 \; \forall e \in E_{G_0} \text{ and } p_e = 0 \text{ otherwise}\}$ is a feasible point.
Let $k_u$ be the number of potential edges incident to node $u$. By requiring $u$'s expected degree to be unchanged, we have $\sum_{v \in \mathcal{N}(u)} p_{uv} = d_u$. Applying Cauchy-Schwarz inequality, we get $\sum_{v \in \mathcal{N}(u)} p_{uv}^2 \geq \frac{1}{d_u + k_u} (\sum_{v \in \mathcal{N}(u)} p_{uv})^2 = \frac{d_u^2}{d_u + k_u}$. Now we take the sum over all nodes to get the following
\begin{multline*}
Var(E) = m - \sum_{i=1}^{m+n_p} p_i^2 = m - \frac{1}{2}\sum_{u} \sum_{v \in \mathcal{N}(u)} p_{uv}^2 \\
\leq m - \frac{1}{2}\sum_{u}\frac{d_u^2}{d_u+d_k} \leq m - \frac{1}{2}\frac{(\sum_{u}d_u)^2}{\sum_{u}(d_u+k_u)} = \frac{mn_p}{m+n_p} \nonumber
\end{multline*}
where the last equality is again due to Cauchy-Schwarz inequality.
\end{proof}

\subsection{Comparison of schemes}
\label{subsec:compare-model}
Table \ref{tab:model-compare} shows the comparison of schemes we investigate in this work. Only MaxVar and EdgeSwitch satisfy all three properties (1),(2) and (3). The next two propositions quantify the TV of $(k,\epsilon)$-obf and \textit{RandWalk-mod}.

\begin{table}[htb]
 \small
 \centering
 \caption{Comparison of schemes} \label{tab:model-compare}
 \begin{tabular}{|l|c|c|c|c|}
 \hline
 \textbf{Scheme} & Prop \#1 & Prop \#2 & Prop \#3 & Uncertain $\mathcal{A}$\\
 \hline
 \textit{RandWalk-mod} & $\circ$ ($\alpha=0.5$) & $\times$ & $\circ$ & $\circ$\\
 \hline
 \textit{RandWalk} & $\circ$ & $\circ$ & $\times$ & $\circ$\\
 \hline
 \textit{EdgeSwitch} & $\circ$ & $\circ$ & $\circ$ & $\times$\\
 \hline
 $(k,\epsilon)$\textit{-obf} & $\circ$ & $\circ$ & $\times$ & $\circ$\\  
 \hline
 \textit{MaxVar} & $\circ$ & $\circ$ & $\circ$ & $\circ$\\
 \hline
 \textit{Mixture} & \multicolumn{4}{c|}{depends on the mixed scheme} \\
 \hline
 \textit{Partition} & \multicolumn{4}{c|}{depends on the scheme used in subgraphs} \\ 
 \hline

 \end{tabular}
 \end{table}

\begin{proposition}
\label{propos:obf-var}
The expected total variance of $(k,\epsilon)$-obf $TV_{obf}$ is $(m+n_p)(E[r_e]-E[r_e^2])$. The expressions of $E[r_e], E[r_e^2]$ are given in (\ref{eqn:E_re}) and (\ref{eqn:E2_re}).
\end{proposition}
\begin{proof}
In $(k,\epsilon)$-obf, $m$ existing edges are assigned probabilities $1-r_e$ while $n_p$ potential edges are assigned probabilities $r_e$. Therefore, the total variance is $TV_{obf} = m(1-r_e)(1-(1-r_e)) + n_p r_e(1-r_e) = (m+n_p)r_e(1-r_e)$ where $r_e \leftarrow R_\sigma$. Take the expectation of $TV_{obf}$, we get $E[TV_{obf}] = (m+n_p)(E[r_e]-E[r_e^2])$. 

$R_\sigma$ has pdf $f(x) = C\frac{1}{\sigma\sqrt{2\pi}}e^{-x^2/2\sigma^2} \text{ if } x\in [0,1] \text{ and } 0 \text{ otherwise}$. The normalization constant $C = 0.5\text{erf}(1/\sigma\sqrt{2})$ where erf is the error function. Basic integral computations (change of variable and integration by parts) give us the formulas for $E[r_e]$ and $E[r_e^2]$ as follows
\begin{align}
E[r_e] &= \frac{C\sigma}{\sqrt{2\pi}}(1-e^{-1/2\sigma^2}) \label{eqn:E_re}\\
E[r_e^2] &= \frac{C\sigma}{\sqrt{2\pi}}(\frac{\sigma\sqrt{2\pi}}{C}-e^{-1/2\sigma^2}) \label{eqn:E2_re}
\end{align}
\end{proof}
Note that for $\sigma \leq 0.1$, $C \approx 1$ and $e^{-1/2\sigma^2} \approx 0$, so 
\begin{equation}
E[TV_{obf}] \approx (m+n_p)\left(\frac{\sigma}{\sqrt{2\pi}} - \sigma^2\right) 
\end{equation}

\begin{proposition}
\label{propos:randwalk-var}
The total variance of \textit{RandWalk-mod} $TV_{RW}(t)$ at walk-length $t$ is upper bounded by $\frac{m(K_t-m)}{K_t}$ where $K_t$ is the number of non-zeros in $B^{(t)}$. 

For power-law graphs with the exponent $\gamma$, $TV_{RW}^{PL}(\infty) = m-\frac{1}{2}\left[\frac{\zeta(\gamma-2)}{\zeta(\gamma-1)}\right]^2$.
For sparse ER random graphs with $\lambda=np$ constant, $TV_{RW}^{ER}(\infty) = m-\frac{1}{2}(\lambda+1)^2$
\end{proposition}
\begin{proof} The proof uses the same arguments as in Proposition \ref{propos:max-var} and Theorem \ref{theorem:limit}. We omit it due to space limitation. 
\end{proof}

Note that the $K_t$ increases with $t$ and when $t$ is equal to the diameter of $G$, $K_t = n^2$. Therefore, the upper bound of $TV_{RW}(t)$ converges very fast to $m$, compatible with the results in the limit cases of \textit{PL} and \textit{ER} random graphs.

\section{Quantifying Framework}
\label{sec:quantify}
This section describes a generic framework for privacy and utility quantification of anonymization methods.

\subsection{Privacy Measurement}
We focus on structural re-identification attacks under various models of attacker's knowledge as shown in \cite{hay2008resisting}. We quantify the privacy of an anonymized graph as the \textit{sum} of re-identification probabilities of all nodes in the graph. We differentiate \textit{closed-world} from \textit{open-world} adversaries. For example, when a closed-world adversary knows that Bob has three neighbors, this fact is exact. An open-world adversary in this case would learn only that Bob has at least three neighbors. We consider the result of structural query $Q$ on a node $u$ as the node signature $sig_Q(u)$. Given a query $Q$, nodes having the same signatures form an \textit{equivalence class}. So given the true graph $G_0$ and an output anonymized graph $G^*$, the privacy is measured as in the following example.

\begin{example}
Assuming that we have signatures of $G_0$ and signatures of $G^*$ as in Table \ref{tab:ex1}, the re-identification probabilities in $G^*$ of nodes 1,2 are $\frac{1}{3}$, of nodes 4,8 are $\frac{1}{2}$, of nodes 3,5,6,7 are 0s. And the privacy score of $G^*$ is $\frac{1}{3} + \frac{1}{3} + \frac{1}{2} + \frac{1}{2} + 0 + 0 + 0 + 0 = 1.66$. In $G_0$, the privacy score is $\frac{1}{3} + \frac{1}{3} + \frac{1}{3} + \frac{1}{2} + \frac{1}{2} + \frac{1}{3} + \frac{1}{3} + \frac{1}{3} = 3$, equal to the number of equivalence classes.
\end{example}

\begin{table}[htb]
\small
\centering
\caption{Example of node signatures} \label{tab:ex1}
\begin{tabular}{|c|l|}
\hline
\textbf{Graph} &\textbf{Equivalence classes} \\
\hline
$G_0$ & $s_1\{1,2,3\}, s_2\{4,5\}, s_3\{6,7,8\}$ \\
\hline
$G^*$ & $s_1\{1,2,6\}, s_2\{4,7\}, s_3\{3,8\}, s_4\{5\}$\\
\hline
\end{tabular}
\end{table}

We consider two privacy scores in this paper. 
\begin{itemize}
\item $\mathbf{H1}$ score uses node degree as the node signature, i.e. we assume that the attacker know \textit{apriori} degrees of all nodes. 
\item $\mathbf{H2_{open}}$ uses the \textit{set} (not multiset) of degrees of node's friends as the node signature. For example, if a node has 6 neighbors and the degrees of those neighbors are $\{1,2,2,3,3,5\}$, then its signature for $H2_{open}$ attack is $\{1,2,3,5\}$.
\end{itemize}
 Higher-order scores like $H2$ (exact multiset of neighbors' degrees) or $H3$ (exact multiset of neighbor-of-neighbors' degrees) induce much higher privacy scores of the true graph $G_0$ (in the order of $|V|$) and represent less meaningful metrics for privacy. The following proposition claims the \textit{automorphism-invariant} property of structural privacy scores.
 
\begin{proposition}
\label{propos:automorphism}
All privacy scores based on structural queries \cite{hay2008resisting} are automorphism-invariant, i.e. if we find a non-trivial automorphism $G_1$ of $G_0$, the signatures of all nodes in $G_1$ are unchanged.
\end{proposition}
\begin{proof} The proof is trivially based on the definition of graph automorphism. We omit it due to the lack of space.
%
%
\end{proof}

\subsection{Utility Measurement}
\label{subsec:utility-measure}
Following \cite{boldi2012injecting} and \cite{ying2008randomizing}, we consider three groups of statistics for utility measurement: degree-based statistics, shortest-path based statistics and clustering statistics.

\subsubsection{Degree-based statistics}
\begin{itemize}
\item Number of edges: $S_{NE} = \frac{1}{2}\sum_{v\in V} d_v$
\item Average degree: $S_{AD} = \frac{1}{n}\sum_{v\in V} d_v$
\item Maximal degree: $S_{MD} = \max_{v\in V} d_v$
\item Degree variance: $S_{DV} = \frac{1}{n}\sum_{v\in V} (d_v-S_{AD})^2$
\item Power-law exponent of degree sequence: $S_{PL}$ is the estimate of $\gamma$ assuming the degree sequence follows a power-law $\Delta(d) \sim d^{-\gamma}$
\end{itemize}

\subsubsection{Shortest path-based statistics}
\begin{itemize}
\item Average distance: $S_{APD}$ is the average distance among all pairs of vertices that are path-connected.
\item Effective diameter: $S_{EDiam}$ is the 90-th percentile distance among all path-connected pairs of vertices.
\item Connectivity length: $S_{CL}$ is defined as the harmonic mean of all pairwise distances in the graph.
\item Diameter : $S_{Diam}$ is the maximum distance among all path-connected pairs of vertices.
\end{itemize}

\subsubsection{Clustering statistics}
\begin{itemize}
\item Clustering coefficient: $S_{CC} = \frac{3N_{\Delta}}{N_3}$ where $N_{\Delta}$ is the number of triangles and $N_3$ is the number of connected triples.
\end{itemize}

All of the above statistics are computed on sample graphs generated from the uncertain output $\mathcal{G}$. In particular, to estimate shortest-path based measures, we use Approximate Neighbourhood Function (ANF) \cite{palmer2002anf}. The diameter is lower bounded by the longest distance among all-destination bread-first-searches from 1,000 randomly chosen nodes.

\section{Evaluation}
\label{sec:eval}
In this section, our evaluation aims to show the disadvantages of \textit{$(k,\epsilon)$-obf} and \textit{RandWalk/RandWalk-mod} as well as the gap between them. We then illustrate the effectiveness and efficiency of the gap-filling approaches \textit{MaxVar} and \textit{Mixture}. The effectiveness is measured by privacy scores (lower is better) and the relative error of utility (lower is better). The efficiency is measured by the running time. All algorithms are implemented in Python and run on a desktop PC with $Intel^{\circledR}$ Core i7-4770@ 3.4Ghz, 16GB memory. We use \textit{MOSEK}\footnote{http://mosek.com/} as the quadratic solver. 

Three large real-world datasets are used in our experiments \footnote{http://snap.stanford.edu/data/index.html}. \texttt{dblp} is a co-authorship network where two authors are connected if they publish at least one paper together. \texttt{amazon} is a product co-purchasing network where the graph contains an undirected edge from $i$ to $j$ if a product $i$ is frequently co-purchased with product $j$. \texttt{youtube} is a video-sharing web site that includes a social network. The graph sizes $(|V|,|E|)$ of \texttt{dblp}, \texttt{amazon} and \texttt{youtube} are (317080, 1049866), (334863, 925872) and (1134890, 2987624) respectively. We partition \texttt{dblp}, \texttt{amazon} into 20 subgraphs and \texttt{youtube} into 60 subgraphs. The sample size of each test case is 20.

\subsection{$(k,\epsilon)$-obf and RandWalk}
We report the performance of \textit{$(k,\epsilon)$-obf} in Table \ref{tab:k-obfuscation}. We keep the number of potential edges equal to $m$ (default value in \cite{boldi2012injecting}) and vary $\sigma$. We find that the scheme achieves low relative errors only at small $\sigma$. However, privacy scores, especially $H2_{open}$, rise fast (up to 50\% compared to the true graph). This fact incurs high privacy-utility tradeoff as confirmed in Table \ref{tab:compare}.

Table \ref{tab:randwalk} shows the performance similarity between \textit{RandWalk} and \textit{RandWalk-mod} except the case of \texttt{youtube} and for $t=2$ in \texttt{amazon}. Because \textit{RandWalk-mod} satisfies the third constraint, it benefits several degree-based statistics while the existence of selfloops and multiedges does not impact much on shortest-path based metrics. \textit{RandWalk} misses a lot of edges at $t=2$ (see footnote \ref{foot:1} in Section \ref{subsec:model-randwalk}). The remarkable characteristics of random-walk schemes are the very low privacy scores and the high relative errors (lower-bounded around 8 to 10\%). Clearly, there is a gap between high tradeoffs in 
\textit{$(k,\epsilon)$-obf} and high relative errors in \textit{RandWalk} where \textit{MaxVar} and \textit{Mixture} may play their roles.

\begin{table*}[htb]
\small
\centering
\caption{\textit{$(k,\epsilon)$-obf}} \label{tab:k-obfuscation}
\begin{tabular}{r||r|r||r|r|r|r|r|r|r|r|r|r||r}
\hline
$\sigma$ & $H1$& $H2_{open}$& $S_{NE}$ & $S_{AD}$& $S_{MD}$& $S_{DV}$& $S_{CC}$& $S_{PL}$& $S_{APD}$& $S_{ED}$& $S_{CL}$& $S_{Diam}$ & rel.err\\
\hline
\texttt{dblp} & 199 & 125302 & 1049866 & 6.62 & 343 & 100.15 & 0.306 & 2.245 & 7.69 & 9 & 7.46 & 20 &  \\
\hline
0.001 & 72.9 & 40712.1 & 1048153 & 6.61 & 316.0 & 97.46 & 0.303 & 2.244 & 7.74 & 9.4 & 7.50 & 20.0 & \textbf{0.018} \\
0.01  & 41.1 & 24618.2 & 1035994 & 6.53 & 186.0 & 86.47 & 0.294 & 2.248 & 7.82 & 9.5 & 7.59 & 19.8 & \textbf{0.077} \\
0.1   & 19.7 & 7771.4 & 991498 & 6.25 & 164.9 & 64.20 & 0.284 & 2.265 & 8.08 & 10.0 & 7.85 & 20.0 & \textbf{0.128} \\
\hline
\hline
\texttt{amazon} & 153 & 113338 & 925872 & 5.53 & 549 & 33.20 & 0.205 & 2.336 & 12.75 & 16 & 12.10 & 44 & \\
\hline
0.001 & 55.7 & 55655.9 & 924321 & 5.52 & 479.1 & 31.73 & 0.206 & 2.340 & 12.14 & 15.2 & 11.65 & 33.2 & \textbf{0.057} \\
0.01  & 34.5 & 39689.8 & 915711 & 5.47 & 299.7 & 27.18 & 0.220 & 2.348 & 12.40 & 15.6 & 11.91 & 32.4 & \textbf{0.101} \\
0.1   & 19.2 & 16375.4 & 892140 & 5.33 & 253.9 & 21.87 & 0.232 & 2.374 & 12.52 & 15.5 & 12.06 & 31.4 & \textbf{0.144} \\
\hline
\hline
\texttt{youtube} & 978 & 321724 & 2987624 & 5.27 & 28754 & 2576.0 & 0.0062 & 2.429 & 6.07 & 8 & 6.79 & 20 & \textbf{} \\
\hline
0.001 & 157.2 & 36744.6 & 2982974 & 5.26 & 28438 & 2522.6 & 0.0062 & 2.416 & 6.24 & 8.0 & 6.01 & 19.5 & \textbf{0.022} \\ 
0.01  & 80.0 & 22361.7 & 2940310 & 5.18 & 26900 & 2282.6 & 0.0061 & 2.419 & 6.27 & 8.0 & 6.04 & 19.0 & \textbf{0.043} \\ 
0.1   & 23.4 & 5806.9 & 2624066 & 4.62 & 16353 & 970.8 & 0.0070 & 2.438 & 6.59 & 8.1 & 6.36 & 20.4 & \textbf{0.160} \\ 
\hline
\end{tabular}
\end{table*}

\begin{table*}[htb]
\small
\centering
\caption{\textit{RandWalk} and \textit{RandWalk-mod}} \label{tab:randwalk}
\begin{tabular}{r||r|r||r|r|r|r|r|r|r|r|r|r||r}
\hline
$t$ & $H1$& $H2_{open}$& $S_{NE}$ & $S_{AD}$& $S_{MD}$& $S_{DV}$& $S_{CC}$& $S_{PL}$& $S_{APD}$& $S_{ED}$& $S_{CL}$& $S_{Diam}$ & rel.err\\
\hline
\texttt{dblp} & 199 & 125302 & 1049866 & 6.62 & 343 & 100.15 & 0.306 & 2.245 & 7.69 & 9 & 7.46 & 20 &  \\
\hline
(RW) 2 & 10.0 & 4.9 & 1001252 & 6.32 & 309.3 & 86.16 & 0.152 & 2.197 & 7.43 & 9.1 & 7.20 & 19.7 & \textbf{0.094} \\
	3 & 11.8 & 10.9 & 1048129 & 6.61 & 315.4 & 98.04 & 0.107 & 2.155 & 7.08 & 8.7 & 6.88 & 17.8 & \textbf{0.110} \\
	5 & 11.7 & 5.6 & 1049484 & 6.62 & 321.6 & 100.77 & 0.065 & 2.148 & 6.79 & 8.0 & 6.62 & 16.4 & \textbf{0.142} \\
	10 & 11.9 & 2.9 & 1049329 & 6.62 & 329.2 & 103.06 & 0.030 & 2.144 & 6.54 & 8.0 & 6.40 & 14.3 & \textbf{0.171} \\
\hline
(RW-mod) 2 & 11.8 & 4.5 & 1049921 & 6.62 & 327.0 & 105.3 & 0.093 & 2.110 & 7.75 & 9.7 & 7.48 & 23.0 & \textbf{0.109} \\
		3 & 11.9 & 9.4 & 1049877 & 6.62 & 343.3 & 105.1 & 0.071 & 2.117 & 7.32 & 9.0 & 7.10 & 20.4 & \textbf{0.099} \\
		5 & 12.0 & 5.4 & 1049781 & 6.62 & 340.5 & 105.1 & 0.044 & 2.115 & 6.95 & 8.4 & 6.76 & 18.3 & \textbf{0.131} \\
		10 & 11.9 & 2.6 & 1049902 & 6.62 & 340.0 & 105.3 & 0.021 & 2.116 & 6.59 & 8.0 & 6.44 & 16.0 & \textbf{0.164} \\
\hline
\hline
\texttt{amazon} & 153 & 113338 & 925872 & 5.53 & 549 & 33.20 & 0.205 & 2.336 & 12.75 & 16 & 12.10 & 44 & \\
\hline
(RW) 2 & 5.7 & 5.4 & 861896 & 5.15 & 274.9 & 23.11 & 0.148 & 2.337 & 10.70 & 13.8 & 10.19 & 38.7 & \textbf{0.180} \\
	3 & 10.0 & 16.5 & 923793 & 5.52 & 495.6 & 32.72 & 0.113 & 2.282 & 10.33 & 13.1 & 9.87 & 34.1 & \textbf{0.137} \\
	5 & 10.4 & 8.6 & 925185 & 5.53 & 507.7 & 33.52 & 0.080 & 2.276 & 9.45 & 12.1 & 9.07 & 29.6 & \textbf{0.181} \\
	10 & 10.2 & 4.6 & 925748 & 5.53 & 498.1 & 34.37 & 0.046 & 2.273 & 8.55 & 10.5 & 8.25 & 25.7 & \textbf{0.234} \\
\hline
(RW-mod) 2 & 9.8 & 3.2 & 925672 & 5.53 & 255.1 & 37.61 & 0.099 & 2.246 & 12.02 & 15.5 & 11.40 & 43.2 & \textbf{0.139} \\
		3 & 9.9 & 11.2 & 925532 & 5.53 & 535.3 & 37.32 & 0.082 & 2.254 & 10.89 & 14.0 & 10.38 & 37.9 & \textbf{0.134} \\
		5 & 9.7 & 6.0 & 926163 & 5.53 & 522.8 & 37.42 & 0.059 & 2.252 & 9.83 & 12.5 & 9.40 & 33.0 & \textbf{0.185} \\
		10 & 9.9 & 3.3 & 925809 & 5.53 & 491.4 & 37.45 & 0.035 & 2.251 & 8.76 & 11.0 & 8.44 & 28.7 & \textbf{0.238} \\
\hline
\hline
\texttt{youtube} & 978 & 321724 & 2987624 & 5.27 & 28754 & 2576.0 & 0.0062 & 2.429 & 6.07 & 8 & 6.79 & 20 & \textbf{} \\
\hline
(RW) 2 & 13.4 & 1.5 & 2636508 & 4.65 & 19253.8 & 1139.7 & 0.022 & 2.191 & 6.18 & 7.9 & 5.93 & 23.5 & \textbf{0.403} \\
	3 & 23.8 & 17.6 & 2982204 & 5.26 & 26803.6 & 2389.6 & 0.004 & 2.108 & 5.73 & 7.0 & 5.52 & 18.0 & \textbf{0.103} \\
	5 & 24.6 & 8.4 & 2985967 & 5.26 & 26018.7 & 2340.0 & 0.005 & 2.106 & 5.55 & 7.0 & 5.38 & 16.3 & \textbf{0.120} \\
	10 & 21.9 & 1.8 & 2984115 & 5.26 & 24695.8 & 2099.4 & 0.009 & 2.100 & 5.49 & 6.9 & 5.33 & 18.7 & \textbf{0.145} \\
\hline
(RW-mod) 2 & 26.4 & 1.4 & 2987228 & 5.26 & 23829.7 & 2578.5 & 0.018 & 2.053 & 6.27 & 8.0 & 6.02 & 22.1 & \textbf{0.245} \\
		3 & 26.9 & 22.3 & 2988011 & 5.27 & 28611.5 & 2579.7 & 0.005 & 2.077 & 5.75 & 7.2 & 5.54 & 19.0 & \textbf{0.081} \\
		5 & 26.1 & 11.0 & 2987479 & 5.26 & 28619.3 & 2581.4 & 0.005 & 2.076 & 5.61 & 7.0 & 5.44 & 18.3 & \textbf{0.090} \\
		10 & 26.3 & 1.7 & 2987475 & 5.26 & 28432.2 & 2579.9 & 0.008 & 2.073 & 5.58 & 7.0 & 5.41 & 18.8 & \textbf{0.099} \\
\hline
\end{tabular}
\end{table*}

\subsection{Effectiveness of MaxVar}
We assess privacy and utility of \textit{MaxVar} by varying the number of potential edges $n_p$. The results are shown in Table \ref{tab:effectiveness}. As for privacy scores, if we increase $n_p$, we gain better privacy as we allow more edge switches. Due to the expected degree constraints in the quadratic program, all degree-based metrics vary only a little.

We observe the near \textit{linear} relationships between $H1$, $rel.err$ and the number of replaced edges $|E_{G_0}\setminus E_G|$ in Figures \ref{fig:mv_H1}, \ref{fig:mv_rel_err} and near \textit{quadratic} relationship of $H2_{open}$ against $|E_{G_0}\setminus E_G|$ in Fig.\ref{fig:mv_H2}. The ratio of replaced edges in Figures \ref{fig:mv_H1},\ref{fig:mv_H2} and \ref{fig:mv_rel_err} is defined as $\frac{|E_{G_0}\setminus E_G|}{|E_{G_0}|}$.

The runtime of MaxVar consists of time for (1) partitioning $G_0$, (2) adding friend-of-friend edges to subgraphs, (3) solving quadratic subproblems and (4) combining uncertain subgraphs to get $\mathcal{G}$. We report the runtime in Fig.\ref{fig:mv_runtime}. As we can see, the total runtime is in several minutes and the runtime of the partitioning step is almost negligible. Increasing $n_p$ gives rise to the runtime in steps 2,3 and 4 and the trends are nearly linear. The runtime on \texttt{youtube} is three times longer than on the other two datasets, almost linear to their data sizes.

\begin{table*}[htb]
\small
\centering
\caption{Effectiveness of \textit{MaxVar} ($k$ denotes one thousand)} \label{tab:effectiveness}
\begin{tabular}{r||r|r||r|r|r|r|r|r|r|r|r|r||r}
\hline
$n_p$ & $H1$& $H2_{open}$& $S_{NE}$ & $S_{AD}$& $S_{MD}$& $S_{DV}$& $S_{CC}$& $S_{PL}$& $S_{APD}$& $S_{ED}$& $S_{CL}$& $S_{Diam}$ & rel.err\\
\hline
\texttt{dblp} & 199 & 125302 & 1049866 & 6.62 & 343 & 100.15 & 0.306 & 2.245 & 7.69 & 9 & 7.46 & 20 &  \\
\hline
$200k$ & 59.7 & 3257.2 & 1049774 & 6.62 & 342.3 & 100.73 & 0.279 & 2.213 & 7.66 & 9.3 & 7.43 & 19.5 & \textbf{0.017} \\
$400k$ & 40.7 & 744.0 & 1049813 & 6.62 & 343.5 & 101.26 & 0.255 & 2.189 & 7.56 & 9.1 & 7.33 & 18.9 & \textbf{0.030} \\
$600k$ & 32.1 & 325.7 & 1050066 & 6.62 & 343.4 & 101.73 & 0.235 & 2.173 & 7.46 & 9.0 & 7.25 & 17.7 & \textbf{0.045} \\
$800k$ & 29.5 & 199.2 & 1049869 & 6.62 & 345.9 & 102.07 & 0.219 & 2.163 & 7.45 & 9.0 & 7.24 & 17.0 & \textbf{0.056} \\
$1000k$ & 27.0 & 140.7 & 1049849 & 6.62 & 345.4 & 102.29 & 0.205 & 2.155 & 7.34 & 9.0 & 7.15 & 17.0 & \textbf{0.064} \\
\hline
\hline
\texttt{amazon} & 153 & 113338 & 925872 & 5.53 & 549 & 33.20 & 0.205 & 2.336 & 12.75 & 16 & 12.10 & 44 & \\
\hline
$200k$ & 30.2 & 2209.1 & 925831 & 5.53 & 551.5 & 33.83 & 0.197 & 2.321 & 12.38 & 16.1 & 11.72 & 40.5 & \textbf{0.022} \\
$400k$ & 22.8 & 452.4 & 925928 & 5.53 & 550.2 & 34.40 & 0.182 & 2.306 & 11.88 & 15.3 & 11.28 & 37.1 & \textbf{0.050} \\
$600k$ & 17.8 & 188.4 & 925802 & 5.53 & 543.9 & 34.79 & 0.167 & 2.296 & 11.60 & 15.0 & 11.04 & 36.9 & \textbf{0.066} \\
$800k$ & 17.2 & 118.8 & 925660 & 5.53 & 550.0 & 35.11 & 0.154 & 2.289 & 11.33 & 14.4 & 10.81 & 34.5 & \textbf{0.087} \\
$1000k$ & 15.2 & 82.4 & 925950 & 5.53 & 551.8 & 35.43 & 0.142 & 2.282 & 11.13 & 14.1 & 10.62 & 31.8 & \textbf{0.105} \\
\hline
\hline
\texttt{youtube} & 978 & 321724 & 2987624 & 5.27 & 28754 & 2576.0 & 0.0062 & 2.429 & 6.07 & 8 & 6.79 & 20 & \textbf{} \\
\hline
$600k$ & 114.4 & 4428.8 & 2987898 & 5.27 & 28759 & 2576 & 0.0065 & 2.373 & 6.19 & 7.8 & 5.97 & 18.6 & \textbf{0.030} \\ 
$1200k$ & 84.2 & 1419.2 & 2987342 & 5.26 & 28754 & 2576 & 0.0064 & 2.319 & 6.02 & 7.2 & 5.82 & 17.9 & \textbf{0.042} \\ 
$1800k$ & 71.4 & 814.4 & 2987706 & 5.27 & 28745 & 2577 & 0.0062 & 2.287 & 5.97 & 7.1 & 5.78 & 17.2 & \textbf{0.049} \\ 
$2400k$ & 65.3 & 595.5 & 2987468 & 5.26 & 28749 & 2577 & 0.0060 & 2.265 & 5.96 & 7.1 & 5.77 & 16.6 & \textbf{0.056} \\
$3000k$ & 62.8 & 513.7 & 2987771 & 5.27 & 28761 & 2578 & 0.0058 & 2.251 & 5.89 & 7.1 & 5.71 & 16.4 & \textbf{0.062} \\
\hline
\end{tabular}
\end{table*}

\begin{figure}
\centering
\begin{minipage}{.22\textwidth}
	\centering
	\includegraphics[height=1.4in]{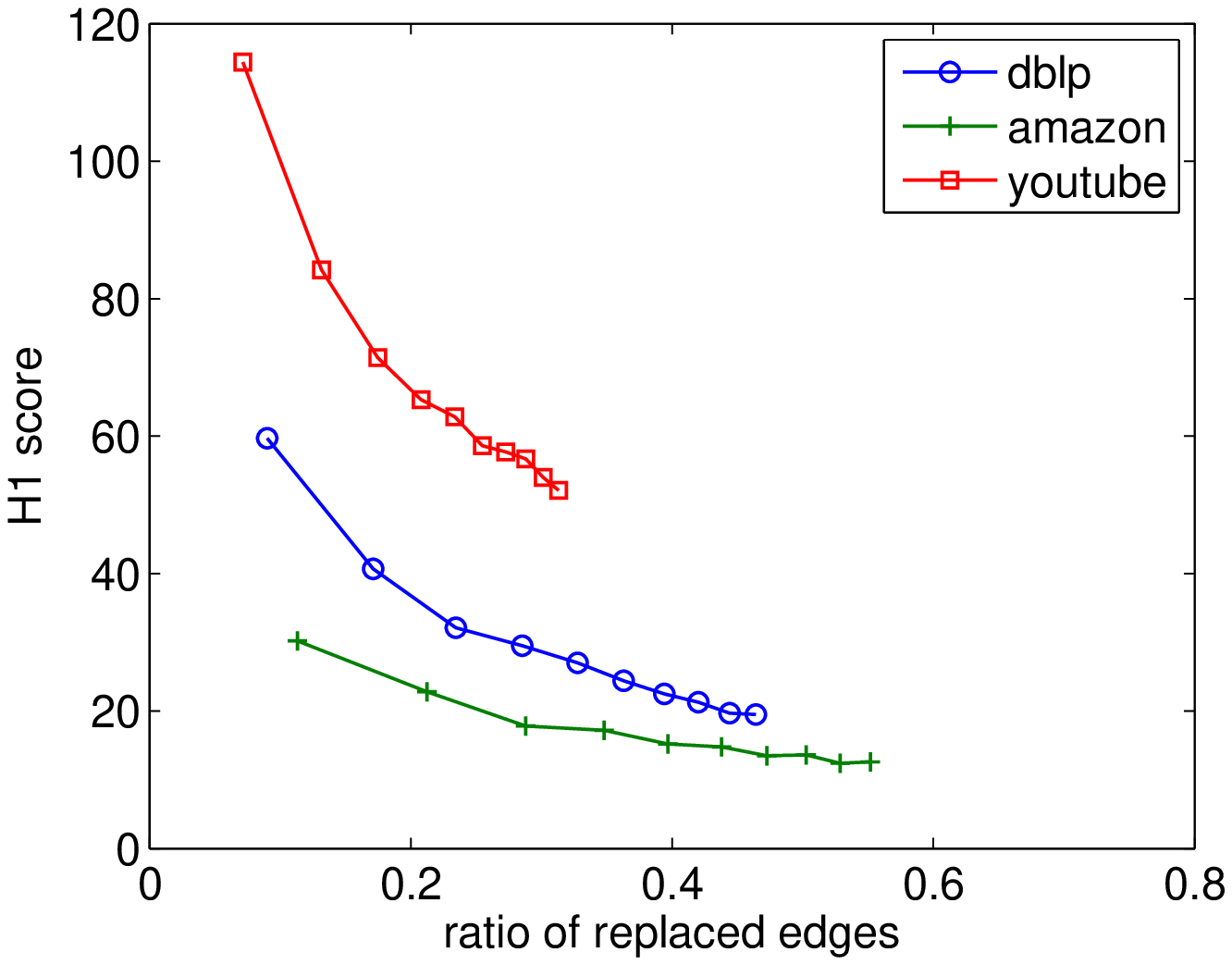}
	\captionof{figure}{$H1$ of MaxVar}
	\label{fig:mv_H1}
\end{minipage}%
\hfill
\begin{minipage}{.22\textwidth}
	\centering
	\includegraphics[height=1.4in]{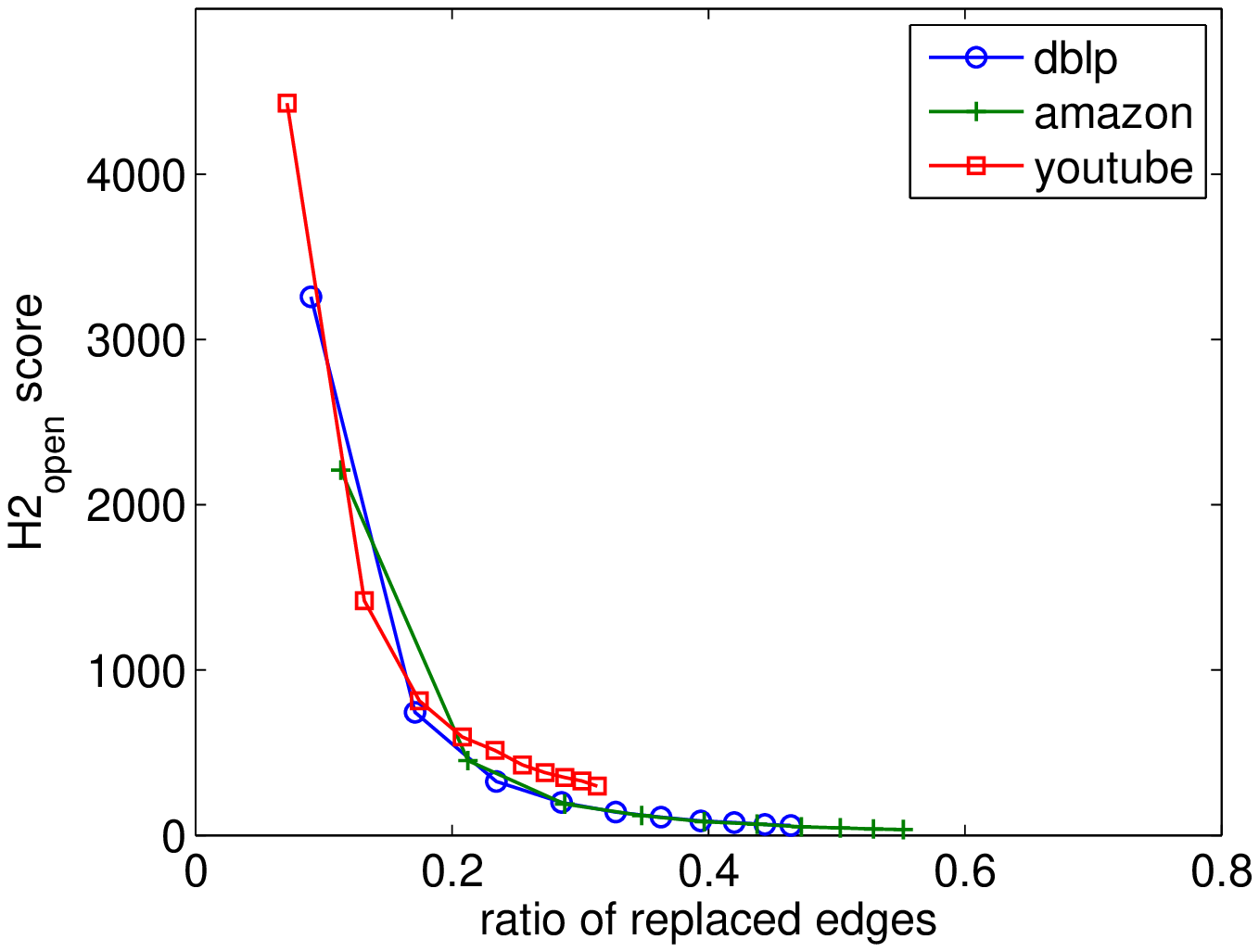}
	\captionof{figure}{$H2_{open}$ of MaxVar}
	\label{fig:mv_H2}
\end{minipage}
\vspace{-1.0em}
\end{figure}

\begin{figure}
\begin{minipage}{.22\textwidth}
		\centering
		\includegraphics[height=1.4in]{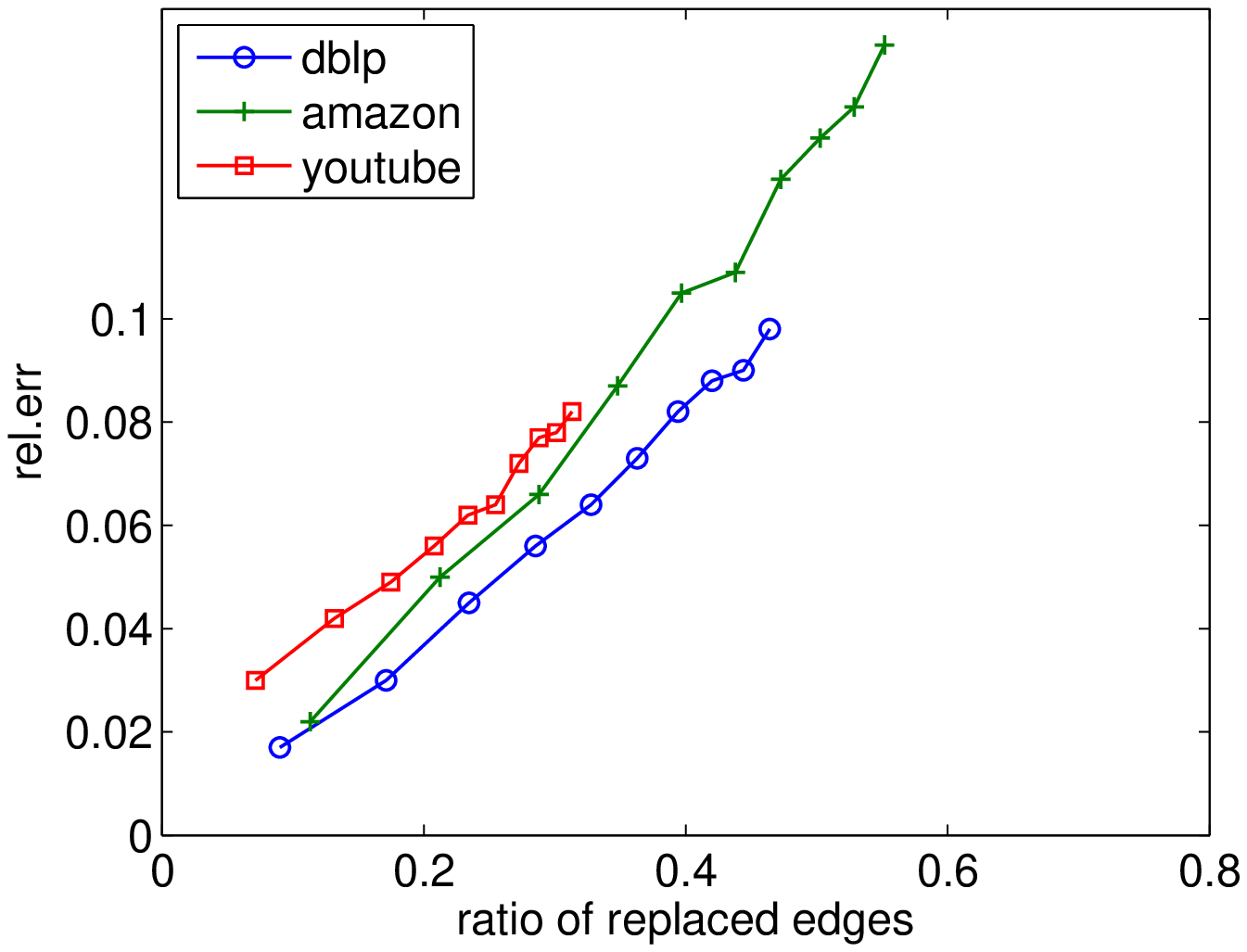}
		\captionof{figure}{Rel.error of MaxVar}
		\label{fig:mv_rel_err}	
\end{minipage}%
\hfill
\begin{minipage}{.22\textwidth}
	\centering
	\includegraphics[height=1.4in]{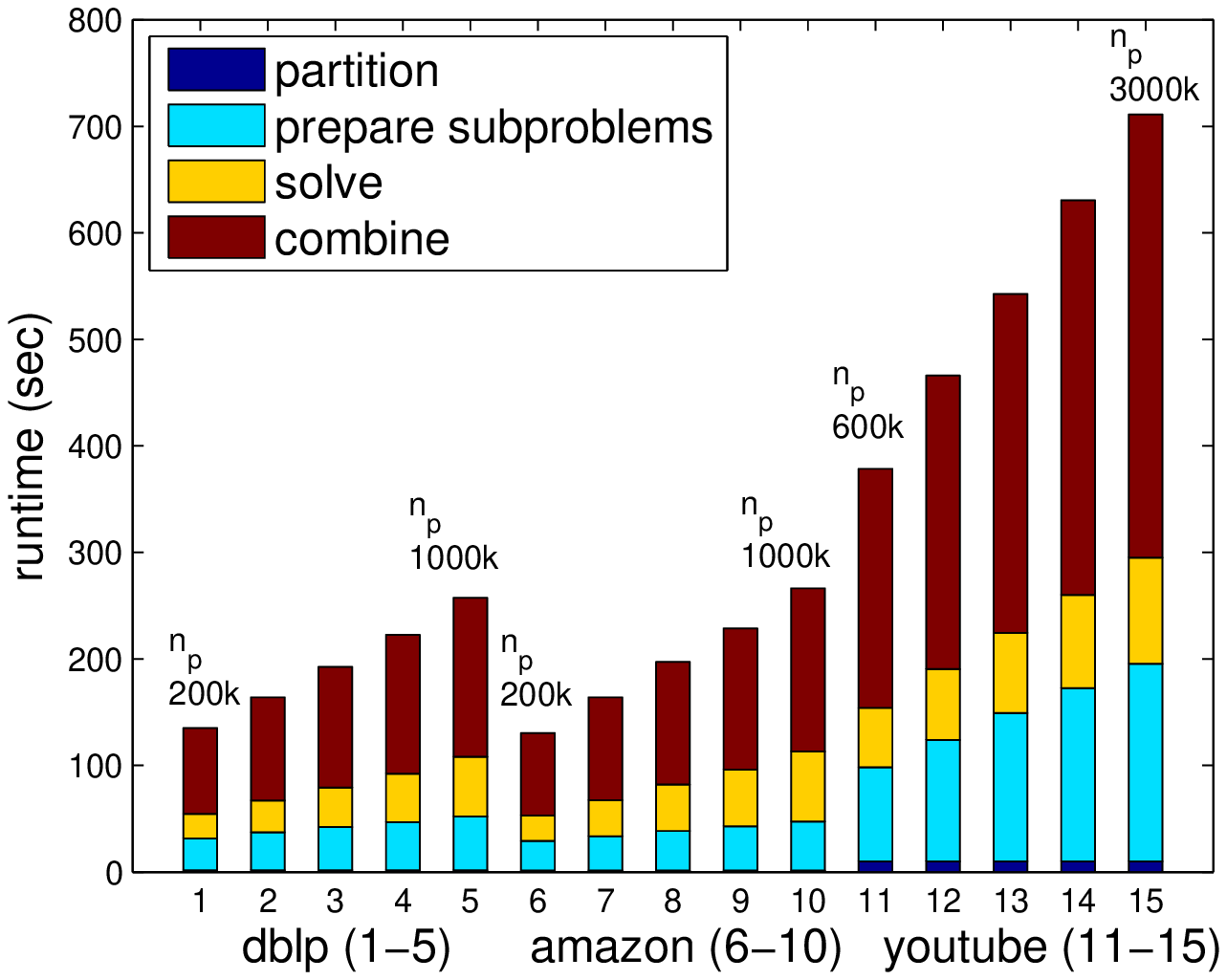}
	\captionof{figure}{Runtime of MaxVar}
	\label{fig:mv_runtime}
\end{minipage}
\vspace{-1.0em}
\end{figure}

\subsection{Comparative Evaluation}

\begin{figure}
\centering
\begin{minipage}{.22\textwidth}
	\centering
	\includegraphics[height=1.4in]{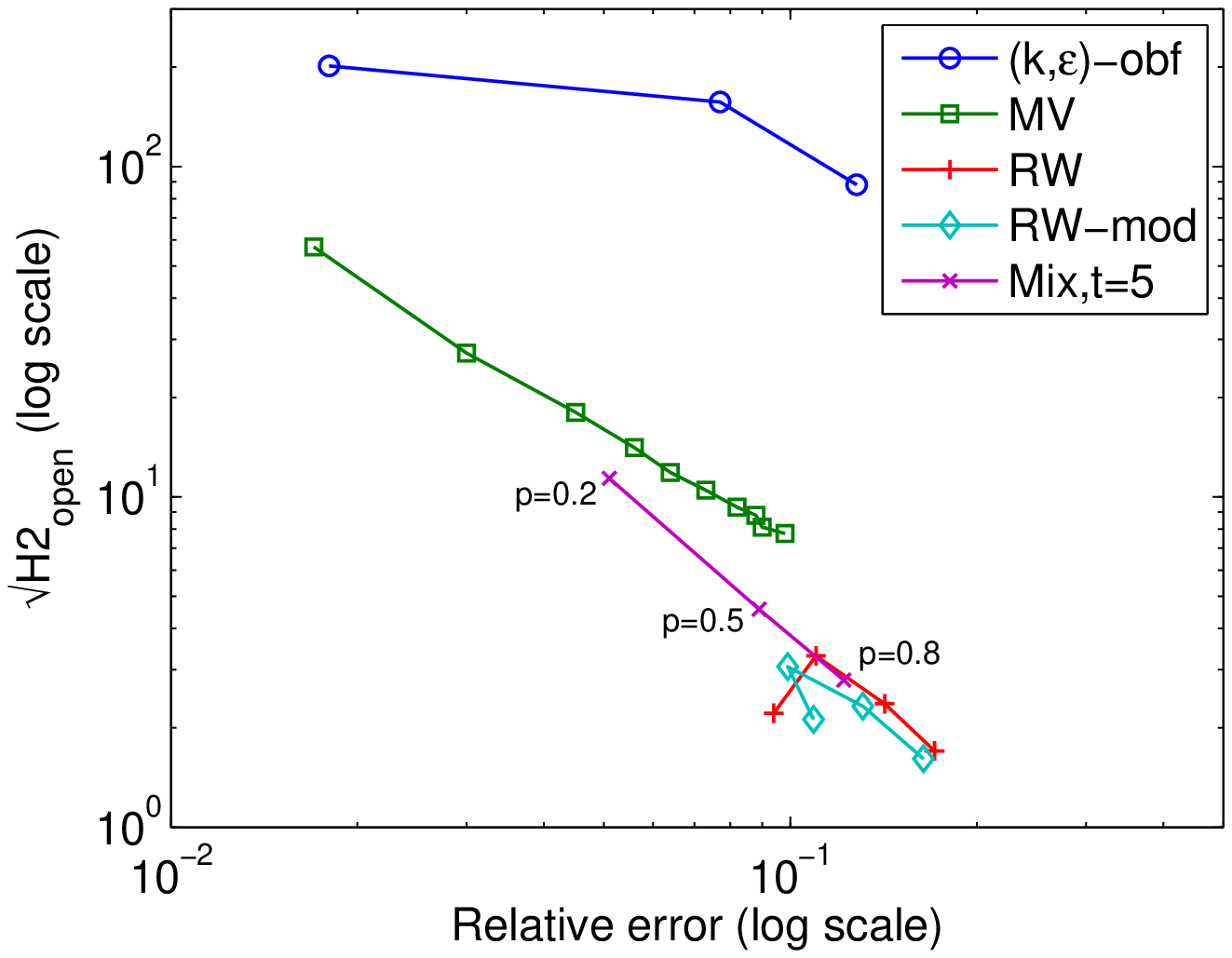}
	\captionof{figure}{Tradeoff in dblp}
	\label{fig:tradeoff-dblp}
\end{minipage}%
\hfill
\begin{minipage}{.22\textwidth}
	\centering
	\includegraphics[height=1.4in]{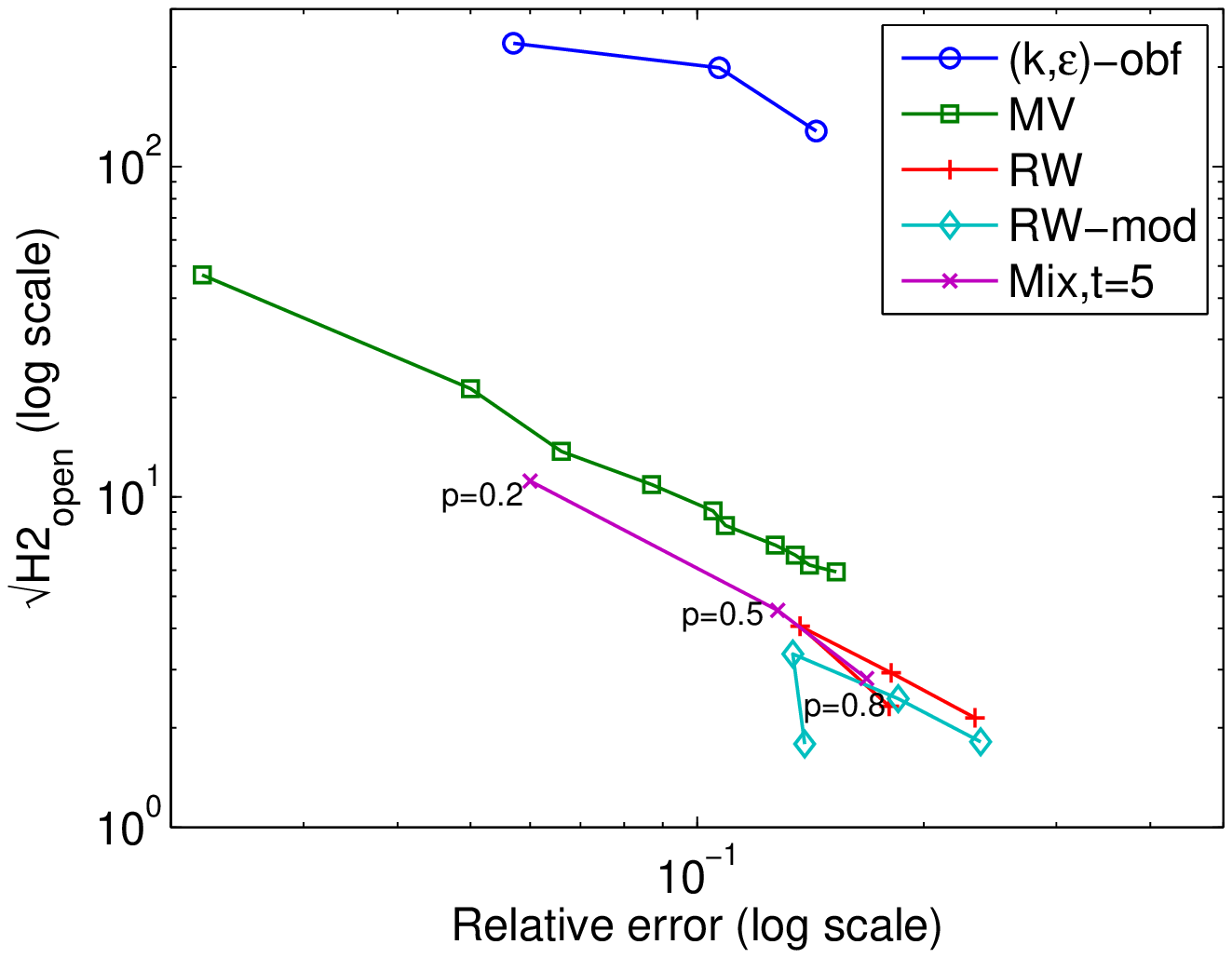}
	\captionof{figure}{Tradeoff in amazon}
	\label{fig:tradeoff-amazon}
\end{minipage}
\vspace{-1.0em}
\end{figure}

\begin{figure}
\centering
\begin{minipage}{.22\textwidth}
	\centering
	\includegraphics[height=1.4in]{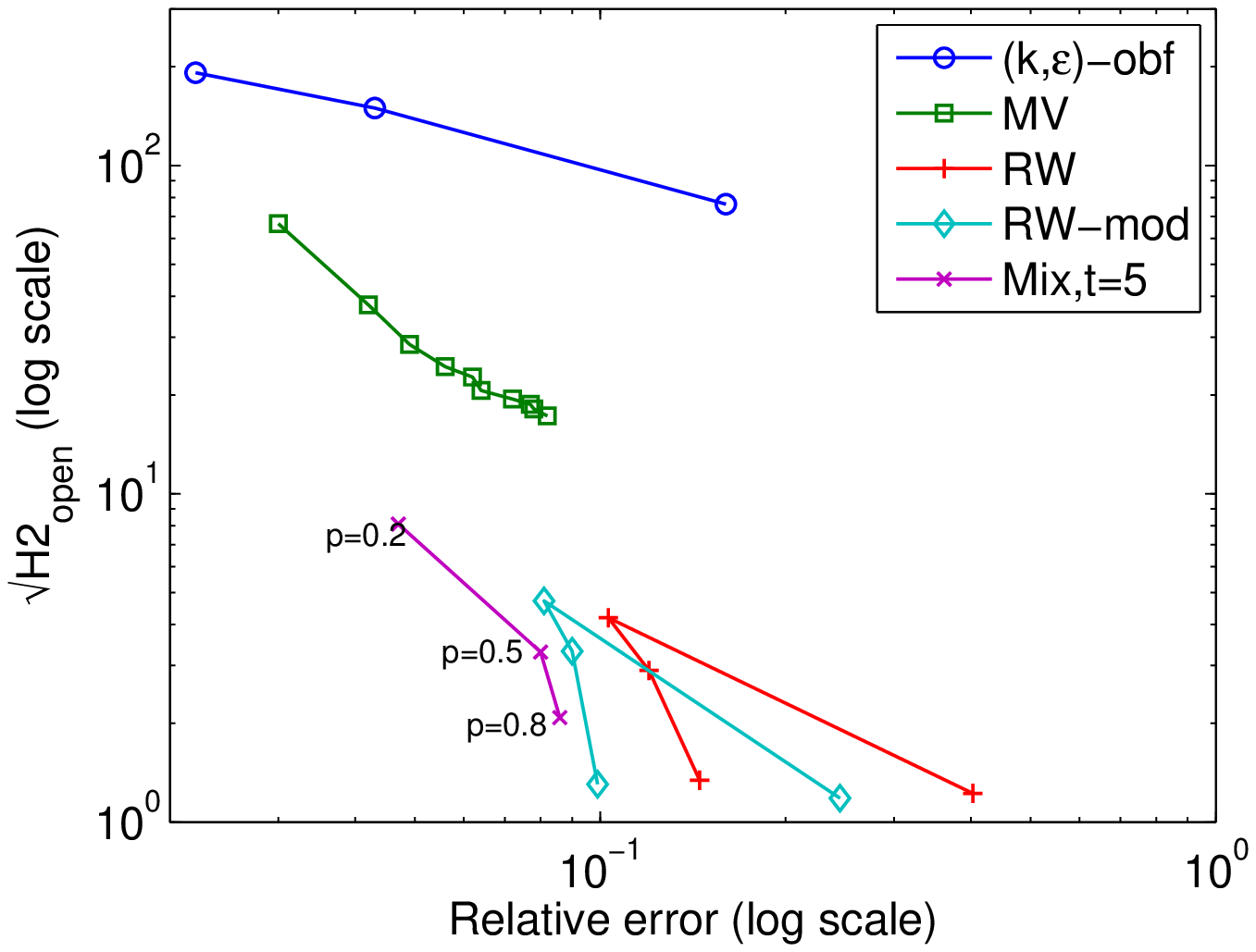}
	\captionof{figure}{Tradeoff in youtube}
	\label{fig:tradeoff-youtube}
\end{minipage}%
\hfill
\begin{minipage}{.22\textwidth}
	\centering
	\includegraphics[height=1.4in]{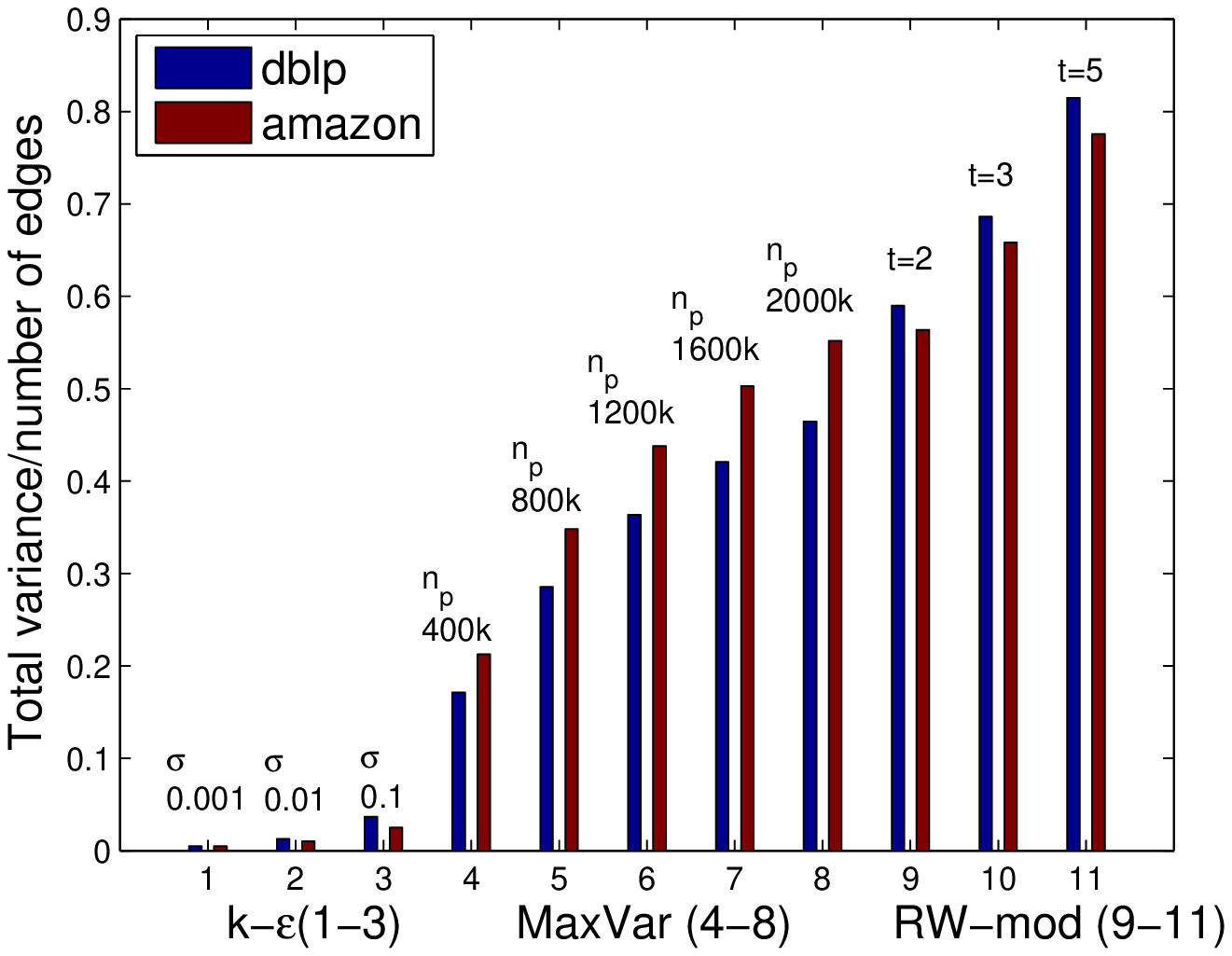}
	\captionof{figure}{Total Variance (TV)}
	\label{fig:total-variance}
\end{minipage}
\vspace{-1.0em}
\end{figure}

Table \ref{tab:compare} shows comparisons between \textit{MaxVar}, $(k,\epsilon)$-obf and \textit{RandWalk/RandWalk-mod}. The column \textit{tradeoff} is $\sqrt{H2_{open}} \times rel.err$ as we conjecture the quadratic and linear behavior of $H2_{open}$ and $rel.err$ respectively (Figures \ref{fig:mv_H2} and \ref{fig:mv_rel_err}). We omit the column $H1\times rel.err$ because they are almost equal for all schemes considered in this work. Clearly, MaxVar gains better privacy-utility tradeoffs than $(k,\epsilon)$-obf, but worse than \textit{RandWalk}, \textit{RandWalk-mod}. However, MaxVar has its own merit as a gap-filling solution. Figures \ref{fig:tradeoff-dblp},\ref{fig:tradeoff-amazon} and \ref{fig:tradeoff-youtube} show that while RandWalk, RandWalk-mod have the best tradeoffs, they suffer from high lower bounds for utility. In other words, if the dataset allows higher privacy risk for better utility (lower rel.err) then the usage of two random walk based solutions may be limited. The simple solution \textit{Mixture} also fills the gap. We omit EdgeSwitch due to its worst tradeoffs.

In addition to the re-identification scores $H1$ and $H2_{open}$, we also compute $\epsilon$ for $k \in \{30, 50,100\}$ to have a fair comparison with $(k,\epsilon)$-obf. Table \ref{tab:compare} shows that \textit{MaxVar} has the best $(k,\epsilon)$ scores. The number of potential edges used in MaxVar could be 20\% of $|E_{G_0}|$, much less than that of $(k,\epsilon)$-obf (100\% for $c = 2$ \cite{boldi2012injecting}). \textit{MaxVar} and \textit{RandWalk/RandWalk-mod} have $|E_{G_0}\setminus E_G| \simeq |E_G\setminus E_{G_0}|$ and these two quantities are higher than those of $(k,\epsilon)$-obf where the number of edges is preserved only at small $\sigma$. RandWalk and RandWalk-mod do not have many edges preserved due to their rewiring nature. $|E_{G_0}\setminus E_G|$ increases slowly in MaxVar because the edges in $G_0$ always have positive probabilities. Fig \ref{fig:total-variance} compares the normalized total variance (i.e. divided by $|E_{G_0}|$) of three schemes. Again, MaxVar is between $(k,\epsilon)$-obf and RandWalk-mod.

\begin{table*}[htb]
\small
\centering
\caption{\textit{MaxVar} vs. \textit{$(k,\epsilon)$-obf}, \textit{RandWalk} and \textit{RandWalk-mod} (lower tradeoff is better)} \label{tab:compare} 
\begin{tabular}{r||r|r||r|r||r|r|r||r|r}
\hline
 & \multicolumn{7}{c}{PRIVACY} & \multicolumn{2}{c}{UTILITY} \\
\hline
graph & $H1$& $H2_{open}$& $|E_{G_0}\setminus E_G|$ & $|E_G\setminus E_{G_0}|$& $\epsilon (k=30)$& $\epsilon (k=50)$& $\epsilon (k=100)$& rel.err & tradeoff \\
\hline
\texttt{dblp} & 199 & 125302 &  &  & 0.00238 & 0.00393 & 0.00694 & & \\
\hline
$\sigma = 0.001$ & 72.9 & 40712.1 & 6993.0 & 5280.2 & 0.00039 & 0.00122 & 0.00435 &  0.018 & \textbf{3.61}\\
$\sigma = 0.01$ & 41.1 & 24618.2 & 19317.3 & 5444.9 & 0.00051 & 0.00062 & 0.00082 & 0.077 & \textbf{12.03}\\
$\sigma = 0.1$ & 19.7 & 7771.4 & 65285.1 & 6916.8 & 0.00179 & 0.00199 & 0.00245 & 0.128 & \textbf{11.33}\\
\hline
(MV) $n_p=200k$ & 59.7 & 3257.2 & 94508.0 & 94416.5 & 0.00033 & 0.00077 & 0.00152 & 0.017 & \textbf{0.99}\\
(MV) $n_p=600k$ & 32.1 & 325.7 & 246155.6 & 246355.3 & 0.00017 & 0.00029 & 0.00085 & 0.045 & \textbf{0.82}\\
\hline
(RW) $t=2$ & 10.0 &	4.9 & 615966.2 & 567352.2 & 0.00318 & 0.00439 & 0.00789 & 0.094 & \textbf{0.21}\\
(RW) $t=5$ & 11.7 &	5.6 & 754178.1 & 753796.3 & 0.00271 & 0.00386 & 0.00689 & 0.142 & \textbf{0.34}\\
\hline
(RW-mod) $t=2$ & 11.8 & 4.5 & 719361.2 & 719416.5 & 0.00073 & 0.00135 & 0.00252 & 0.109 & \textbf{0.23}\\
(RW-mod) $t=5$ &  12.0 & 5.4 & 784872.3 & 784786.8 & 0.00057 & 0.00113 & 0.00228 & 0.131 & \textbf{0.30}\\
\hline
\hline
\texttt{amazon}  & 153 & 113338 &  &  & 0.00151 & 0.00218 & 0.00456 & & \\
\hline
$\sigma = 0.001$ & 55.7 & 55655.9 & 6158.9 & 4607.4 & 0.00048 & 0.00119 &  0.00293 & 0.065 & \textbf{13.40}\\
$\sigma = 0.01$ & 34.5 & 39689.8 & 14962.0 & 4801.3 & 0.00038 & 0.00052 & 0.00066  & 0.114 & \textbf{21.33}\\
$\sigma = 0.1$ & 19.2 & 16375.4 & 39382.6 & 5650.3 & 0.00068 & 0.00102 & 0.00190  & 0.145 & \textbf{18.46}\\
\hline
(MV) $n_p=200k$ & 30.2 & 2209.1 &  104800.9 & 104759.9 & 0.00023 & 0.00032 & 0.00065  & 0.022 & \textbf{1.03}\\
(MV) $n_p=600k$ & 17.8 & 188.4 &266603.7 & 266533.7 & 0.00015 & 0.00023 & 0.00047 & 0.066 & \textbf{0.91}\\
\hline
(RW) $t=2$ &  5.7 & 5.4 & 649001.0 & 585025.5  & 0.00213 & 0.00338 & 0.00550 & 0.180 & \textbf{0.42}\\
(RW) $t=5$ & 10.4 & 8.6 & 629961.8 & 629274.9  & 0.00146 & 0.00239 & 0.00423 & 0.181 & \textbf{0.53}\\
\hline
(RW-mod) $t=2$ & 9.8 & 3.2 &  725440.1 & 725239.9 & 0.00048 & 0.00073 & 0.00133 & 0.139 & \textbf{0.25}\\
(RW-mod) $t=5$ & 9.7 & 6.0 &  671694.2 & 671985.4 & 0.00038 & 0.00058 & 0.00137 &  0.185 & \textbf{0.45}\\
\hline
\hline
\texttt{youtube} & 978 & 321724 &   &  & 0.00291 & 0.00402 & 0.00583  & & \\
\hline
$\sigma = 0.001$ & 157.2 & 36744.6 & 19678.5 & 15028.5 & 0.00143 & 0.00232 & 0.00421 & 0.022 & \textbf{4.28}\\
$\sigma = 0.01$ & 80.0 & 22361.7 & 62228.6 & 14914.3 & 0.00060 & 0.00105 &  0.00232 & 0.043 & \textbf{6.38}\\
$\sigma = 0.1$ & 23.4 & 5806.9 & 378566.0 & 15007.5 & 0.00038 & 0.00052 & 0.00074 & 0.160 & \textbf{12.20}\\
\hline
(MV) $n_p=600k$ & 114.4 & 4428.8 &  213097.3 & 213371.4 & 0.00047 & 0.00063 & 0.00108 & 0.030 & \textbf{2.00}\\
(MV) $n_p=1800k$ & 71.4 & 814.4 &  521709.9 & 521791.6 & 0.00040 & 0.00052 & 0.00090 & 0.049 & \textbf{1.38}\\
\hline
(RW) $t=2$ & 13.4 & 1.5 &  2836169.3 & 2485053.4  & 0.00319 & 0.00425 & 0.00623 & 0.403 & \textbf{0.50}\\
(RW) $t=5$ & 24.6 & 8.4 &  2468068.6 & 2466411.3  & 0.00304 & 0.00408 & 0.00598 & 0.120 & \textbf{0.35}\\
\hline
(RW-mod) $t=2$ & 26.4 & 1.4 & 2863112.1 & 2862716.3 & 0.00159 & 0.00226 & 0.00355 & 0.245 & \textbf{0.29}\\
(RW-mod) $t=5$ & 26.1 & 11.0 & 2467414.9 & 2467269.5 & 0.00153 & 0.00322 & 0.00159 & 0.090 & \textbf{0.30}\\
\hline
\end{tabular}
\end{table*}

\section{Conclusion}
We provide a generalized view of graph anonymization based on the semantics of edge uncertainty. Via the model of uncertain adjacency matrix with the constraint of unchanged expected degree for all nodes, we analyze recently proposed schemes and explain why there exists a gap between them by comparing the total degree variance. We propose \textit{MaxVar}, a novel anonymization scheme exploiting two key observations: maximizing the total degree variance while keeping the expected degrees of all nodes unchanged and using nearby potential edges. We also investigate an elegant \textit{Mixture} approach that together with \textit{MaxVar} fill the gap between \textit{$(k,e)$-obf} and \textit{RandWalk}. Furthermore, we promote the usage of incorrectness measure for privacy assessment in a new quantifying framework rather than Shannon entropy and min-entropy (k-anonymity). The experiments demonstrate the effectiveness of our methods. Our work may incite several directions for future research including (1) novel constructions of uncertain graphs based on the uncertain adjacency matrix (2) deeper analysis on the privacy-utility relationships in MaxVar (e.g. explaining the near linear and near quadratic curves) (3) study on directed and bipartite graphs.



%

\bibliographystyle{abbrv}
\bibliography{uncertain-graph-short}

\appendix
\section{Proofs}

\subsection{Proof of theorem \ref{theorem:limit}}
\label{proof:limit}

\begin{proof}
For power-law graphs, the node degree distribution is $P(k)=\frac{k^{-\gamma}}{\zeta(\gamma)}$. The number of selfloops $n_{sl}^{PL}$ in $B_{RW}^{\infty}$ is the sum of elements on the main diagonal.
\begin{multline*}
n_{sl}^{PL} = \frac{1}{2m}\sum_{i=1}^{n} d_i^2 = \frac{1}{2m}\sum_{k=1}^{\infty} k^2nP(k) = \frac{n}{nE(k)}\sum_{k=1}^{\infty} k^2P(k) \\
=  \frac{1}{E(k)}\sum_{k=1}^{\infty} \frac{k^{-(\gamma-2)}}{\zeta(\gamma)} = \frac{\zeta(\gamma)}{\zeta(\gamma-1)} \frac{\zeta(\gamma-2)}{\zeta(\gamma)} = \frac{\zeta(\gamma-2)}{\zeta(\gamma-1)} \nonumber
\end{multline*}
To prove that there is no multiedge in $B_{RW}^{\infty}$ we show that all elements in $B_{RW}^{\infty}$ are less than 1. This is equivalent to show $d_{max} < \sqrt{2m}$. We use the constraint that the number of nodes with degree $d_{max}$ must be at least 1, i.e. $\frac{n d_{max}^{-\gamma}}{\zeta(\gamma)} \geq 1 \leftrightarrow d_{max} \leq (n/\zeta(\gamma))^{1/\gamma}$. Because $\zeta(\gamma) > 1$ and we consider $\gamma > 2$ in social networks, $(n/\zeta(\gamma))^{1/\gamma} < \sqrt{n}$. Meanwhile, $\sqrt{2m} = \sqrt{n\frac{\zeta(\gamma - 1)}{\zeta(\gamma)}} > \sqrt{n}$ due to the fact that $\zeta(\gamma)$ is monotonically decreasing. So we conclude $n_{me}^{PL} = 0$.

For sparse ER random graphs, we have $P(k) \rightarrow e^{-\lambda} \frac{\lambda^k}{k!}$. The number of selfloops $n_{sl}^{ER}$ is
\begin{multline*}
n_{sl}^{ER} = \frac{1}{2m}\sum_{i=1}^{n} d_i^2 = \frac{1}{2m}\sum_{k=1}^{\infty} k^2nP(k) = \frac{n}{nE(k)}\sum_{k=1}^{\infty} k^2P(k) \\
= \frac{1}{\lambda}E(k^2) = \frac{1}{\lambda}(E(k)^2 + Var(k)) = \frac{1}{\lambda} (\lambda^2 + \lambda) = \lambda + 1
\end{multline*}
Similar to the case of PL graphs, we show that $d_{max} < \sqrt{2m}$ where $d_{max} = \max_k n e^{-\lambda}\frac{\lambda^k}{k!} \geq 1 = \max_k \frac{k!}{\lambda^k} \leq  n e^{-\lambda}$. Using the basic facts $k^{k/2} \leq k!$ and $k > \lambda$ we get $\frac{k^{k/2}}{\lambda^k} \leq n e^{-\lambda} < n$, so $k < n^{2/k}\lambda^2 < \sqrt{n\lambda} = \sqrt{2m}$ as long as $n$ is sufficiently large and $\lambda \geq 4$. So we conclude $n_{me}^{ER} = 0$.
\end{proof}

\subsection{Proof of theorem \ref{theorem:edit-distance}}
\label{proof:edit-distance}

\begin{proof}
We prove the result by induction. 

When $k=1$, we have two cases of $G_1$: $E_{G_1}=\{e_1\}$ and $E_{G_1}=\emptyset$. For both cases, $Var[D(\mathcal{G}_1,G_1)] = p_1(1-p_1)$, i.e. independent of $G_1$.

Assume that the result is correct up to $k-1$ edges, i.e. $Var[D(\mathcal{G}_{k-1},G_{k-1})] = \sum_{i=1}^{k-1} p_i(1-p_i)$ for all $G_{k-1} \sqsubseteq \mathcal{G}_{k-1}$, we need to prove that it is also correct for $k$ edges. We use the subscript notations $\mathcal{G}_k, G_k$ for the case of $k$ edges. We consider two cases of $G_k$: $e_k \in G_k$ and $e_k \notin G_k$.

\textit{Case 1.} The formula for $Var[D(\mathcal{G}_k, G_k)]$ is
\begin{multline*}
Var[D(\mathcal{G}_k,G_k)] = \sum_{G'_k \sqsubseteq \mathcal{G}_k} Pr(G'_k) [D(G'_k,G_k) - E[D(\mathcal{G}_k,G_k)]]^2\nonumber \\
=\sum_{e_k \in G'_k} Pr(G'_k) [Dk - E[D_k]]^2 + \sum_{e_k \notin G'_k} Pr(G'_k) [D_k - E[D_k]]^2 \nonumber
\end{multline*}

The first sum is $\sum_{G'_{k-1} \sqsubseteq \mathcal{G}_{k-1}} p_k Pr(G'_{k-1})[D_{k-1} - E[D_{k-1}] - (1-p_k)]^2$.

The second sum is $\sum_{G'_{k-1} \sqsubseteq \mathcal{G}_{k-1}} (1-p_k) Pr(G'_{k-1})[D_{k-1} - E[D_{k-1}] + p_k)]^2$.

Here we use shortened notations $D_k$ for $D(G'_k,G_k)$ and $E[D_k]$ for $E[D(\mathcal{G}_k,G_k)]$.

By simple algebra, we have $Var[D(\mathcal{G}_k,G_k)] = Var[D(\mathcal{G}_{k-1},G_{k-1})] + q_k(1-q_k) = \sum_{i=1}^{k} p_i(1-p_i)$. 

\textit{Case 2.} similar to the Case 1.
\end{proof}

%

\end{document}